\documentclass[journal]{IEEEtran}

\usepackage{filecontents}
\usepackage[noadjust]{cite}

\ifCLASSINFOpdf
\else
\fi
\usepackage{amsmath}
\usepackage{amssymb}
\usepackage{mathtools}
\usepackage{enumitem}

\usepackage{graphicx}
\usepackage{array}
\usepackage{bm}
\usepackage{circuitikz}
\usepackage{tikz}
\usepackage{siunitx} 
\usepackage[hypcap]{caption}
\usetikzlibrary{arrows}
\newcommand*{\MakeBox}[1]{\makebox[1.5em][r]{#1}}%

\DeclarePairedDelimiter\floor{\lfloor}{\rfloor}
\newcommand{\RNum}[1]{\uppercase\expandafter{\romannumeral #1\relax}}
\usepackage{adjustbox}

\usepackage{algorithmic}
\usepackage{array}
\usepackage{makecell}
\usepackage{multirow}
\newcolumntype{M}[1]{>{\centering\arraybackslash}m{#1}}
\usepackage{amsthm} 
\newtheorem{theorem}{Theorem} 

\usepackage[noadjust]{cite} 

\begin{document}

\title{Orthogonal and Non-Orthogonal Signal Representations Using New Transformation Matrices Having NPM Structure}
\author{Shaik~Basheeruddin~Shah,~\IEEEmembership{Student~Member,~IEEE,}
~Vijay~Kumar~Chakka,~\IEEEmembership{Senior~Member,~IEEE,}
~and~Arikatla~Satyanarayana~Reddy
\thanks{Shaik Basheeruddin Shah and Vijay Kumar Chakka are with the Department of Electrical Engineering, Shiv Nadar University, India (e-mail: bs600@snu.edu.in, Vijay.Chakka@snu.edu.in).}
\thanks{Arikatla Satyanarayana Reddy is with the Department of Mathematics, Shiv Nadar University, India (e-mail: satyanarayana.reddy@snu.edu.in).}}
\maketitle

\begin{abstract}
In this paper, we introduce two types of real-valued sums known as Complex Conjugate Pair Sums (CCPSs) denoted as CCPS$^{(1)}$ and CCPS$^{(2)}$, and discuss a few of their properties.
Using each type of CCPSs and their circular shifts, we construct two non-orthogonal Nested Periodic Matrices (NPMs).
As NPMs are non-singular, this introduces two non-orthogonal transforms known as Complex Conjugate Periodic Transforms (CCPTs) denoted as CCPT$^{(1)}$ and CCPT$^{(2)}$. We propose another NPM, which uses both types of CCPSs such that its columns are mutually orthogonal, this transform is known as Orthogonal CCPT (OCCPT). 
After a brief study of a few OCCPT properties like periodicity, circular shift, etc., we present two different interpretations of it.
Further, we propose a Decimation-In-Time (DIT) based fast computation algorithm for OCCPT (termed as FOCCPT), whenever the length of the signal is equal to $2^v,\ v{\in}\mathbb{N}$.
The proposed sums and transforms are inspired 
by Ramanujan sums and Ramanujan Period Transform (RPT).
Finally, we show that the period (both divisor and non-divisor) and frequency information of a signal can be estimated using the proposed transforms with a significant reduction in the computational complexity over Discrete Fourier Transform (DFT).
\end{abstract}

\begin{IEEEkeywords}
Complex exponential, Ramanujan sum, CCPS, DFT, RPT, NPM, OCCPT, FOCCPT, Frequency estimation.
\end{IEEEkeywords}

\IEEEpeerreviewmaketitle

\section{Introduction}
\IEEEPARstart{I}{n} general, information of a finite length signal
like period/frequency is not manifested in a recognizable fashion from the signal itself. 
Such information can be extracted by representing the signal using certain types of bases. This makes signal representation as one of the fundamental problems in signal processing.

In literature, different application specific representations are proposed like DFT, Discrete Cosine Transform, Discrete Sine Transform,  etc., \cite{oppenheim1975digital,Proakis}. 
In $1918$, mathematician Srinivasa Ramanujan introduced an integer-valued summation known as {\em Ramanujan sum} \cite{Ramanujan}. It has some interesting properties like periodicity, orthogonality, etc., which attracted many researchers to use it for various applications \cite{Planat, Samadi, Sugavaneswaran}.
In $2014$, P. P. Vaidyanathan introduced a signal representation known as RPT by using Ramanujan sums and their circular shifts \cite{PP1}, \cite{PP2}.
It has certain periodicity properties, which are useful for period estimation \cite{PP2}. 
Moreover, the presence of an integer basis makes it computationally efficient.
Due to this, RPT has been used in many period estimation applications \cite{7527160, BrainStimuli, Monaural, 8335717, 7952298}.
Generalizing RPT, in \cite{7109930}, the authors introduced a family of full rank square matrices known as Nested Periodic Matrices (NPMs).
Members of the NPM family include 
Natural Basis matrix, Hadamard matrix, DFT matrix, RPT matrix, etc.

Let $p_1,p_2,\dots,p_m$ be all the divisors of $N$ (where $N{\in}\mathbb{N}$) and $s_{p_i}$ be a $\varphi(p_i)$ dimensional subspace consists of $p_i$ periodic signals, where $\varphi(.)$ is an Euler's totient function.
As $\sum_{{p_i}|N}^{}\varphi(p_i) = N$, an $N^{th}$ order NPM is constructed by providing the basis for all $\{s_{p_i}\}_{i=1}^m$. The commonality between different NPMs is that they span the same subspaces ($\{s_{p_i}\}_{i=1}^m$) by providing alternate bases.
Since ${p_i}|N$, $s_{p_i}$ is known as divisor subspace.
In the NPM family, the performance of RPT and DFT is good in period estimation \cite{7109930}.

Two discrete-time signals with different discrete frequencies may correspond to the same period. For example, two sinusoidal signals with discrete frequencies $\frac{2{\pi}}{5}$ and $\frac{2{\pi}(3)}{5}$ are periodic with period $5$.
Hence, extracting the frequency information is equally important as period extraction
and it is crucial in many applications. 
For instance, brain signals exhibit periodic nature in response to external visual stimuli \cite{BrainStimuli,661272}.
Further, they are classified into five frequency sub-bands (alpha, beta, theta, delta, and gamma) over a periodic duration.
Electrocardiogram (ECG) is another signal that exhibits periodic nature and we can segment it in QRS complex segment, ST segment, etc., 
based on the frequency range \cite{Elgendi}.

One of the limitations of RPT is that it does not provide the frequency information \cite{Shah},
whereas DFT gives period 
as well as frequency information with high computational complexity.
In this paper, we address this problem of signal representation using certain new transformation matrices following NPM structure, such that both period and frequency information can be extracted with less computational complexity.
\subsection{Contributions of This Paper}

Ramanujan sums are generated by adding certain complex exponential sequences satisfying some periodicity property \cite{Ramanujan}. The RPT matrix is constructed by using Ramanujan sums and their circular shifts \cite{PP1}, \cite{PP2}. 
Inspired by this, we have proposed two types of real-valued summations known as Complex Conjugate Pair Sum of type-$1$ (CCPS$^{(1)}$) and Complex Conjugate Pair Sum of type-$2$ (CCPS$^{(2)}$).
A few of their properties are also discussed.
Two new NPMs are constructed by using each type of CCPSs and their circular shifts as a basis for $\{s_{p_i}\}_{i=1}^m$.
Further, a finite length signal is represented by using these NPMs. The corresponding transforms are named as Complex Conjugate Periodic Transforms.
Based on the type of summation used, it is denoted as either CCPT$^{(1)}$ or CCPT$^{(2)}$, we have shown that both of these transforms are non-orthogonal. 
One of our previous works \cite{Shah} contains the basic idea of generating CCPS$^{(1)}$, CCPT$^{(1)}$ and their application, which are discussed in brief in this paper.

Another NPM is proposed by using both types of complex conjugate pair sums as a basis for $\{s_{p_i}\}_{i=1}^m$, such that the columns of the matrix are mutually orthogonal.
The corresponding transform is named as Orthogonal CCPT (OCCPT).
The proposed transforms may have applications in communication, image processing, control applications \cite{7842433,503278,8626481,8544037}.
Several important properties of OCCPT like  periodicity, circular convolution, etc., are stated and proved, and its 
relation with DFT is also discussed.
In addition, we have presented two different interpretations of OCCPT such that the period and frequency information is explicitly available in each interpretation.
If $N=2^v$, $v{\in}\mathbb{N}$, a 
DIT algorithm is proposed to reduce the computational complexity of 
OCCPT known as Fast OCCPT (FOCCPT).  
We have shown that this algorithm requires $Nlog_2(N)-N+1$ real multiplications and $2Nlog_2(N)-7\left(\frac{N}{2}\right)+5$ real additions for a given $x(n){\in}\mathbb{R}^N$.

We have proved that we can estimate the divisor period and its corresponding frequency information of a signal using the proposed transforms.
DFT is a standard transform that can serve for the same purpose,
so, the computational complexity of proposed transforms is compared with DFT.
If $N{\neq}2^v$, the complexities of OCCPT and DFT are compared using the direct method, though there exist efficient algorithms for DFT \cite{Burrus,Singleton}.
The following list of conclusions are drawn regarding the computational complexity between DFT and OCCPT:
\begin{itemize}
\item If $x(n){\in}\mathbb{C}^N$ and $N=2^v$, then their complexities are comparable with each other.
\item If $x(n){\in}\mathbb{R}^N$ and $N=2^v$, then Fast Fourier Transform (FFT) requires $2Nlog_2(N)$ real
multiplications and $3Nlog_2(N)$ real additions. This is approximately $50\%$ higher than FOCCPT complexity.
\item If $x(n){\in}\mathbb{C}^N$ and $N{\neq}2^v$, then the number of multiplications/additions required for OCCPT is approximately $50\%$ lower than DFT.
\item If $x(n){\in}\mathbb{R}^N$ and $N{\neq}2^v$, then DFT and OCCPT require the same complexity due to the complex conjugate symmetry of DFT coefficients.
\end{itemize}

So far we have dealt with NPMs, which are constructed using the basis of divisor subspaces. In general, there are scenarios where we deal with non-divisor subspaces.
One such scenario is the estimation of non-divisor periods and their corresponding frequencies using a dictionary based approach,
where the proposed CCPT dictionaries have approximately $75\%$ less computational complexity over the DFT dictionary.
Although, 
RPT dictionary has the computational advantage over CCPT dictionaries due to its integer-valued basis, but it does not provide the frequency information. 
As an example, we have evaluated the 
performance of proposed transforms on an ECG signal by considering the problem of R peak delineation.
The results are compared with DFT and RPT.
\subsection{Outline and Notations}
The structure of this paper is as follows: The NPM structure and its properties are briefly reviewed in Section \RNum{2}. 
The process of RPT matrix construction from the DFT matrix is explained in Section \RNum{3}. 
Section \RNum{4} discusses two types of complex conjugate pair sums and their properties. 
Then, two non-orthogonal transforms are introduced
by using each type of summation and its circular shifts in Section \RNum{5}.
In Section \RNum{6}, we introduce an orthogonal transform known as OCCPT and discuss a few of its properties.
A fast computation algorithm for OCCPT is proposed in Section \RNum{7}.
Later, in Section \RNum{8}, the proposed transforms are compared with DFT and RPT.
Conclusions are drawn in Section \RNum{9}.
The following notations are used throughout the paper:
\begin{enumerate}[label={\textbullet\protect\MakeBox{\arabic{enumi}\alph*.}}, align=left, before={\stepcounter{enumi}}, leftmargin=5pt]
\item[$\mathbb{N},\mathbb{Z}:$] Set of  natural numbers and integers respectively.
\item[$\mathbb{R}, \mathbb{C}:$] Set of  real numbers and complex numbers respectively.
\item[$\mathbf{A^T}:$] The transpose of a matrix $\mathbf{A}$.
\item[$\mathbf{A^H}:$] The conjugate transpose of a matrix $\mathbf{A}$.
\item[$r(\mathbf{A}):$]  The rank of a matrix $\mathbf{A}$.
\item[$(a,b):$] Greatest common divisor of  $a$ and $b$.
\item[$lcm:$] Least common multiple.
\item[$\floor*{a}:$] The greatest integer less than or equal to $a,$ where $a\in \mathbb{R}.$
\item[$a|b:$] $a$ divides $b$ and $a{\nmid}b$ denotes $a$ does not divide $b$.
\item[$\varphi:$] Euler's totient function, defined as 
$\varphi(n) = \sum\limits_{i=1}^{n}\floor*{\frac{1}{(i,n)}}$. Since $(i,n) = (n-i,n)$, $\varphi(n)$ is even for $n{\geq}3$.
\item[$M_{m,n}(\mathbb{C}):$] Set of all $m\times n$ matrices with entries from complex numbers. If $m=n$, $M_{m,n}(\mathbb{C})=M_n(\mathbb{C}).$
\item[$D_N:$] The set of all positive divisors of $N.$ We assume $D_N=\{p_1,p_2,\ldots,p_m\},$ where $1=p_1<p_2<\dots<p_m=N.$
\item[$((n))_N:$] Indicates $n \pmod N$.
\item[$U_n:$] $\{k{\in}\mathbb{N}|1\le k\le n,(k,n)=1\}.$ Hence the cardinality of $U_n$, i.e.,  $\# U_n$ is equal to $\varphi(n)$.
\item[$\hat{U}_n:$] $\{k{\in}\mathbb{N}{\mid}1{\leq}k{\leq}\floor*{\frac{{n}}{2}},(k,n)=1,n>2\}$. Hence $\#{\hat{U}_n} = \frac{\varphi({n})}{2}.$ Let $\tilde{U}_{n}={U}_{n}-\hat{U}_{n}$.
\item[$\left<a,b\right>:$]Dot product between $a$ and $b$.
\end{enumerate}
\section{Nested Periodic Matrix Structure}
\label{Nested Periodic Matrix Structure}
\textit{Definition 1:}\cite{7109930}\label{def:NPM}
 A matrix $\mathbf{M}\in M_N(\mathbb{C})$ is said to have {\em NPM structure} (or simply NPM) if 
\begin{equation}
\mathbf{M} = [\mathbf{M_{p_1}}, \mathbf{M_{p_2}},\dots, \mathbf{M_{p_i}},\dots, \mathbf{M_{p_m}}],
\label{NPB}
\end{equation}
satisfies the following three properties:
\begin{itemize}
\item $\mathbf{M_{p_i}}{\in}M_{N,\varphi(p_i)}(\mathbb{C})$ and $r(\mathbf{M_{p_i}}) = \varphi(p_i)$.  The size of $\mathbf{M}$ is $N{\times}N$ by invoking
 $\sum\limits_{{p_i}|N}\varphi(p_i) = N$ \cite{PP1}.
\item $\mathbf{M}$ is a full rank matrix.
\item Each column in $\mathbf{M_{p_i}}$ is a $p_i$ periodic sequence.
\end{itemize}

Since $\mathbf{M}$ is a non-singular matrix, the columns of $\mathbf{M}$ form a basis  for $\mathbb{C}^{N}$, known as {\em Nested Periodic Basis} (NPB). 
As a consequence, any finite $N$-length signal can be represented as a linear combination of columns of $\mathbf{M}.$ 
Let $s_{p_i}$ be the subspace spanned by the columns of $\mathbf{M_{p_i}}$, then 
\begin{enumerate}                   
\item The period of every element in $s_{p_i}$ is exactly $p_i$ \cite{7109930}.
\item Let $x(n) = \sum\limits_{i=1}^{M}x_{l_i}(n)$, where the period of $x_{l_i}(n)$ is equal to $l_i$.
In general the period of $x(n)$ is a divisor of $lcm(l_1,\dots,l_M)$. But if $x_{l_i}(n)\in s_{l_i}$, then the period of $x(n)$ is exactly equal to $lcm(l_1,\dots,l_M)$ \cite{7109930}.   \end{enumerate}
As an NPM construction involves the basis of divisor subspaces $\left(\{s_{p_i}\}_{i=1}^m\right)$, it is useful to extract the divisor periods information of a signal \cite{PP2}.

Let $\mathbf{A},\mathbf{B}\in M_N(\mathbb{C})$ denote DFT and RPT matrices respectively. 
Both $\mathbf{A}$ and $\mathbf{B}$ are NPMs \cite{7109930}. 
In the following sections, we imitate the construction procedure of $\mathbf{B}$ from $\mathbf{A}$ for constructing three more matrices $\mathbf{C},\mathbf{D}$ and $\mathbf{E}$ of order $N.$ 
Further, we show that $\mathbf{C},\mathbf{D}$ and $\mathbf{E}$ are NPMs.
First, we recall $\mathbf{A}$ and construction of $\mathbf{B}$ from $\mathbf{A}$ in the following section.
\section{RPT Matrix Construction From DFT Matrix}
\subsection{DFT Matrix}
\label{DFT} 
Let $S_{N,k}(n) = e^{\frac{j2{\pi}kn}{N}},\ 0{\leq}n{\leq}N-1$, then  
\begin{equation}
[\mathbf{A}]_{N{\times}N}=[S_{N,0}(n),S_{N,1}(n),\dots,S_{N,N-1}(n)].
\label{dftMtrx}
\end{equation} 
The period of $k^{th}$ column in $\mathbf{A}$ is equal to $\frac{N}{(k,N)}$, {\it i.e.,} divisor of $N.$ 
For each $p_i{\in}D_N$, the number of columns in $\mathbf{A}$ having period exactly equal to $p_i$ is $\#U_{p_i}$ $=$ $\varphi(p_i)$ \cite{Shah}.
If $U_{p_i} = \{k_1,k_2,\dots,k_{\varphi(p_i)}\}$, then form a sub-matrix $\mathbf{A_{p_i}}$ of $\mathbf{A}$ as given below:
\begin{equation}
\begin{aligned}
&[\mathbf{A_{p_i}}]_{N{\times}\varphi(p_i)} ={[\mathbf{\hat{A}_{p_i}},\mathbf{\hat{A}_{p_i}}\dots,\mathbf{\hat{A}_{p_i}}]^\mathbf{T}},
\text{ where}\\
\mathbf{\hat{A}_{p_i}}&={[{S}_{p_i,k_1}(n),{S}_{p_i,k_2}(n),\dots,{S}_{p_i,k_{\varphi(p_i)}}(n)]}_{{p_i}{\times}\varphi(p_i)}.
\label{Rpi}
\end{aligned}
\end{equation}
So, $\mathbf{A_{p_i}}$ is obtained by repeating $\mathbf{\hat{A}_{p_i}}$ periodically $\frac{N}{p_i}$ times.
 As $\sum_{p_i|N}\varphi(p_i) = N$, by constructing $\mathbf{\hat{A}_{p_i}}$ for every $p_i\in D_N$, we can build an $N\times N$ transformation matrix $\mathbf{A}$, whose columns are permutations of columns of $\mathbf{A}$ given in \eqref{dftMtrx}.
 Using the orthogonality and periodicity properties of $S_{N,k}$, one can check that $\mathbf{A}$ is an NPM \cite{7109930}.
In the following sections, we provide different alternative matrices to $\mathbf{\hat{A}_{p_i}}$ such as $\mathbf{\hat{B}_{p_i}}$, $\mathbf{\hat{C}_{p_i}}$, $\mathbf{\hat{D}_{p_i}}$ and $\mathbf{\hat{E}_{p_i}}$ to construct $\mathbf{B_{p_i}}$, $\mathbf{C_{p_i}}$, $\mathbf{D_{p_i}}$ and $\mathbf{E_{p_i}}$ followed by the construction of $\mathbf{B}$, $\mathbf{C}$, $\mathbf{D}$ and $\mathbf{E}$ respectively.
Now before proceeding further, we prove the following:
\begin{theorem}
Let ${x(n)=\sum_{i=1}^{M}S_{N,k_i}(n)}$, where $M{\leq}N$, ${0{\leq}{k_i},n{\leq}N-1}$ and the values $\{k_i\}_{i=1}^M$ are unique.
If $\mathbf{G_N^M}\in M_N(\mathbb{C})$ is a circulant matrix, whose first column is $x(n)$ and the remaining columns are the circular downshift of the previous columns, then  $r\big(\mathbf{G_N^M}\big) = M$.
\label{Th1}
\end{theorem}
\begin{proof}
The given circulant matrix $\mathbf{G_N^M}$ can be decomposed as $\mathbf{G_N^M} =\mathbf{B}_{N{\times}M}\mathbf{B^H}_{M{\times}N}$, where
\footnotesize
\begin{equation}
\nonumber
\begin{aligned}
\mathbf{B} &= \begin{bmatrix}
S_{N,k_1}(0)&
S_{N,k_2}(0)&
\dots&S_{N,k_M}(0) \\
S_{N,k_1}(1)&
S_{N,k_2}(1)&
\dots&S_{N,k_M}(1)  \\
\vdots& \vdots  & \ddots & \vdots\\
S_{N,k_1}(N-1)&
S_{N,k_2}(N-1)&
\dots&S_{N,k_M}(N-1)  
\end{bmatrix}.
\end{aligned}
\end{equation}
\normalsize
Here, the columns of $\mathbf{B}$ are orthogonal,
so $r(\mathbf{B})=M$. As $r(\mathbf{G_N^M}) = r(\mathbf{B})$ \cite{Strang}, this implies $r\big(\mathbf{G_N^M}\big)=M$.
\end{proof}
\subsection{RPT Matrix}
\label{RPT}
If we add all the columns of $\mathbf{\hat{A}_{p_i}}$ given in (\ref{Rpi}), it generates an integer-valued, $p_i$ periodic sequence $c_{p_i}(n)$, known as \textit{Ramanujan sum} \cite{Ramanujan}, \cite{Hardy}.
From \textbf{Theorem} \ref{Th1}, if we construct a ${p_i}{\times}{p_i}$ circulant matrix $\mathbf{G_{p_i}}$ using $c_{p_i}(n)$, then, $r(\mathbf{G_{p_i}})=\varphi(p_i)$. 
So, using $c_{p_i}(n)$ we can build a matrix $\mathbf{\hat{B}_{p_i}}$ as an alternative to $\mathbf{\hat{A}_{p_i}}$ as follows:
\begin{equation}
\mathbf{\hat{B}_{p_i}} ={[c_{p_i}^0(n),c_{p_i}^1(n),\dots,c_{p_i}^{\varphi(p_i)-1}(n)]}_{{p_i}{\times}\varphi(p_i)},
\label{Rpi1}
\end{equation}
where $c_{p_i}^j(n)$ indicates the circular downshift of the sequence $c_{p_i}(n)$ by $j$ times.
Let $\mathbf{B}{\in} M_N(\mathbb{C})$ be the matrix constructed using $\mathbf{\hat{B}_{p_i}}$, by following the similar way  of $\mathbf{A}$ construction from $\mathbf{\hat{A}_{p_i}}$. Then by invoking the orthogonality and periodicity properties of $c_{p_i}(n)$, it is shown in \cite{7109930} that $\mathbf{B}$ satisfies all the NPM properties. 
Here $\mathbf{B}$ is known as RPT matrix \cite{PP2}.

\textit{Remark 1:}
As $(k_j,p_i)=(p_i-k_j,p_i)$, for every complex sequence $S_{p_i,k_j}(n){\in}\mathbf{\hat{A}_{p_i}}$ there exists a complex conjugate sequence $S_{p_i,p_i-k_j}(n){\in}\mathbf{\hat{A}_{p_i}}$. Both together form a complex conjugate pair. So, there are $\frac{\varphi(p_i)}{2}$ complex conjugate pairs in $\mathbf{\hat{A}_{p_i}}$, {\em i.e.,} $\# {\hat{U}_{p_i}}$ \cite{Shah}.   
In \cite{7544641}, the authors introduced a two dimensional subspace spanned by $\{S_{p_i,k}(n),S_{p_i,p_i-k}(n)\}$ for each $k{\in}\hat{U}_{p_i}$, known as {\em Complex Conjugate Subspace (CCS)}, denoted as $v_{{p_i},k}$.
So, $v_{{p_i},k}$ consists of signals having period exactly equal to ${p_i}$ with discrete frequency $\frac{2{\pi}k}{{p_i}}$ (or) $\frac{2{\pi}{({p_i}-k)}}{{p_i}}$.
Let $\hat{U}_{p_i} = \{k_1,k_2,\dots,k_{\frac{\varphi(p_i)}{2}}\}$, then $\mathbf{\hat{A}_{p_i}}$ can be rewritten with permutation of its columns as follows:
\begin{equation}
\footnotesize
\mathbf{\hat{A}_{p_i}} = [\underbrace{{{S}_{p_i,k_1},\ {S}_{p_i,{p_i}-{k_1}}}}_{\text{Basis of } v_{p_i,k_1}},\dots,\ \underbrace{{S}_{p_i,k_{\frac{\varphi(p_i)}{2}}},\ {S}_{p_i,{p_i}-{k_{\frac{\varphi(p_i)}{2}}}}}_{\text{Basis of } v_{p_i,k_\frac{\varphi(p_i)}{2}}}]_{p_i{\times}\varphi(p_i)}.
\label{Rpi2}
\normalsize
\end{equation}

The following section introduces two types of arithmetic sums and their properties, which are used to construct 
alternate bases for CCS.

\section{Complex Conjugate Pair Sums and Their Properties}
In \cite{Shah}, we proposed a real-valued summation by adding each complex conjugate pair known as \textit{Complex Conjugate Pair Sum of type-1 (CCPS$^{(1)}$)}.
Given any $L\in\mathbb{N}$, the CCPS$^{(1)}$ $ \left(c_{L,k}^{(1)}(n)\right)$ is defined as follows:
\begin{equation}
c_{L,k}^{(1)}(n)
 = 2Mcos\left(\frac{2{\pi}{k}n}{L}\right),
\label{CCPS_Def}
\end{equation} 
where
\begin{equation}
     M =\begin{cases}
	\frac{1}{2},& \text{if}\ L=1\ \text{(or)}\ 2\\
    1, & \text{if }{L{\geq}3}
\end{cases},
    \label{M_Value}
\end{equation}
and if $L\geq 3$ then $k\in\hat{U}_L$, otherwise $k=1$.


Similar to addition, subtraction is another arithmetical operation, which can be performed on each complex conjugate pair without changing its periodicity. 
This defines another real-valued sum known as \textit{Complex Conjugate Pair Sum of type-2 (CCPS$^{(2)}$)}, denoted as $c_{L,k}^{(2)}(n)$ and defined as,
\begin{equation}
c_{L,k}^{(2)}(n) = \begin{cases} 1,&\forall n,\ \text{if}\ L=1\\
(-1)^n,&\text{if}\ L=2\\
2sin\big(\frac{2{\pi}{k}n}{L}\big),& \text{if}\ L{\geq} 3,\ k{\in}\hat{U}_L
\end{cases}.
\end{equation}
%
Let $c_{L,k}^{(*)}(n)$ hereafter denotes either $c_{L,k}^{(1)}(n)$ (or) $c_{L,k}^{(2)}(n)$. 
\subsection{Properties}
\subsubsection{Periodicity} As  $c_{L,k}^{(*)}(n+L) = c_{L,k}^{(*)}(n),$ and $c_{L,(k+L)}^{(*)}(n) = c_{L,k}^{(*)}(n)$\footnote{Use the property $(k,L) = (k+L,L)$, to verify $c_{L,(k+L)}^{(*)}(n) = c_{L,k}^{(*)}(n).$},  CCPSs are periodic with respect to both $n$ and $k$ with period $L$. 
If $L$ is even, then  $c_{L,\left(k+\frac{L}{2}\right)}^{(*)}(n) = (-1)^nc_{L,k}^{(*)}(n)$ and $c_{L,k}^{(*)}\left(n+\frac{L}{2}\right) = (-1)^kc_{L,k}^{(*)}(n)$.
\subsubsection{Symmetric}
$c_{L,k}^{(1)}(n)$ and $c_{L,k}^{(2)}(n)$ (for $L{\geq}3$) are even and odd symmetric sequences respectively, with respect to both $k$ and $n$, i.e.,
$c_{L,k}^{(1)}(L-n)=c_{L,(L-k)}^{(1)}(n)=c_{L,k}^{(1)}(n)\text{ and }c_{L,k}^{(2)}(L-n)=c_{L,(L-k)}^{(2)}(n)=-c_{L,k}^{(2)}(n)$.
\subsubsection{DFT of CCPS} For a given $L{\in}\mathbb{N}$ and $l{\in}\hat{U}_L$, 
\begin{equation} \nonumber \begin{aligned} &DFT[c_{L,l}^{(1)}(n)] = C_{L,l}^{(1)}(k) = \begin{cases} L,&\text{if}\ k = l\ \text{(or)}\ L-l\\ 0,&\text{Otherwise} \end{cases},\\ &DFT[c_{L,l}^{(2)}(n)] = C_{L,l}^{(2)}(k) = \begin{cases} -jL,&\text{if}\ k = l\\ jL,&\text{if}\ k = L-l\\ 0,&\text{Otherwise} \end{cases}, \end{aligned} \end{equation} where $0{\leq}k{\leq}L-1$. The above results are obvious from the definition of Complex Conjugate Pair Sums (CCPSs).
\subsubsection{Sum and sum-of-squares} For a given $L>1$ and $l{\in}\hat{U}_L$ $\sum\limits_{n=0}^{L-1}c_{L,l}^{(*)}(n) = 0.$ Given $L{\in}\mathbb{N}$, using the Parseval's relation between $c_{L,l}^{(*)}(n)$ and $C_{L,l}^{(*)}(k)$, we can write:
\begin{equation} \begin{aligned} 
{\sum\limits_{n=0}^{L-1}\left(c_{L,k}^{(*)}(n)\right)^2 = \frac{1}{L}\sum\limits_{k=0}^{L-1}\left(|C_{L,l}^{(*)}(k)|\right)^2 = 2LM,\ }\normalsize 
\end{aligned} \end{equation} 
where $M$ is defined in (\ref{M_Value}).
\subsubsection{Orthogonality} Given $L = lcm(L_1,L_2)$, ${L_1}{\in}{\mathbb{N}}$, ${L_2{\in}{\mathbb{N}}}$,
$k_1{\in}\hat{U}_{L_1}$,
$k_2{\in}\hat{U}_{L_2}$,
${l_1}\in{\mathbb{Z}}$ and ${l_2}\in{\mathbb{Z}}$, then we can prove the following theorem: 
\begin{theorem}
Any two $L$ length CCPSs$^{(1)}$ (or) CCPSs$^{(2)}$ and their circular shifts are mutually orthogonal, {\em i.e.},
\begin{equation}
\begin{aligned}
&\sum\limits_{n=0}^{L-1}c_{L_1,k_1}^{(*)}(n-{l_1})c_{L_2,k_2}^{(*)}
(n-{l_2})\\& = 2L{M}cos\bigg(\frac{2{\pi}{k_1}({l_1-l_2})}{L_1}\bigg)\delta({L_1}-{L_2})\delta({k_1}-{k_2}).
\end{aligned}
\label{Orth1}
\end{equation}
\label{Th}
\end{theorem}
\begin{proof}
Given in the appendix by assuming $c_{L_1,k_1}^{(*)} = c_{L_1,k_1}^{(1)}$.
\end{proof}
\begin{theorem} 
For a given ${L_1}{\geq}3$ and ${L_2}{\geq}3$, 
both CCPS$^{(1)}$ and CCPS$^{(2)}$ are orthogonal to each other, {\em i.e.,}
\begin{equation}
\begin{aligned}
&\sum\limits_{n=0}^{L-1}c_{L_1,k_1}^{(1)}(n-{l_1})c_{L_2,k_2}^{(2)}
(n-{l_2})\\&=2Lsin\bigg(\frac{2{\pi}{k_1}({l_1-l_2})}{L_1}\bigg)\delta({L_1}-{L_2})\delta({k_1}-{k_2}).
\end{aligned}
\end{equation}
\label{Th3}
\end{theorem}
The above theorem can be proved using the same approach used to prove \textbf{Theorem \ref{Th}}. If $L_1<3$ and $L_2<3$, then $c_{L_1,k_1}^{(1)}(n) =c_{L_2,k_2}^{(2)}(n)$, so \textbf{Theorem \ref{Th3}} satisfies \textbf{Theorem \ref{Th}}. 

Now we discuss, how these summations and their properties are used to
construct the basis of CCS, followed by the construction of $\mathbf{C},\mathbf{D}$, and $\mathbf{E}$ matrices.
\section{New Nested Periodic Matrices}
From \textbf{Theorem} \ref{Th1}, if we construct a ${p_i}{\times}{p_i}$ circulant matrix $\mathbf{G_{p_i,k}^{(1)}}$ using $c_{p_i,k}^{(1)}(n)$, then, $r(\mathbf{G_{p_i,k}^{(1)})}=2$.
Let $c_{{p_i},k}^{(*)j}$ indicate the circular downshift of the sequence ${c}_{{p_i},k}^{(*)}$ by $j$ times.
As $c_{{p_i},k}^{(1)}$ and $c_{{p_i},k}^{(1)1}$ are linearly independent, 
\textit{the first two columns of $\mathbf{G_{{p_i},k}^{(1)}}$ act as a basis for CCS \cite{Shah}.}
So, the matrix $\mathbf{\hat{C}_{p_i}}$ (an alternative to $\mathbf{\hat{A}_{p_i}}$ given in \eqref{Rpi2}) built using this new basis is 
\begin{equation}
\nonumber
\mathbf{\hat{C}_{p_i}} = [\underbrace{{c}_{{p_i},k_1}^{(1)},\ {c}_{{p_i},k_1}^{(1)1}}_{\text{Basis of } v_{p_i,k_1}},\dots,\ \underbrace{{c}_{{p_i},k_{\frac{\varphi(p_i)}{2}}}^{(1)},\ {c}_{{p_i},k_{\frac{\varphi(p_i)}{2}}}^{(1)1}}_{\text{Basis of } v_{p_i,k_\frac{\varphi(p_i)}{2}}}]_{p_i{\times}\varphi(p_i)}.
\label{Rpi3}
\end{equation}
It is shown in \cite{Shah} that the $\mathbf{C}$ built by using $\mathbf{\hat{C}_{p_i}}$ is an NPM. 
Then any $N$ length signal $\mathbf{x}$ can be represented/synthesized as
\begin{equation}
\begin{aligned}
\mathbf{x} &= \mathbf{C}\boldsymbol{\beta^{(1)}} = {[\mathbf{C_{p_1}},\dots, \mathbf{C_{p_i}}, \dots, \mathbf{C_{p_m}}]}_{N{\times}N}{\boldsymbol\beta^{(1)}},\\
&\text{where }[\mathbf{C_{p_i}}]_{N{\times}\varphi(p_i)} ={[\mathbf{\hat{C}_{p_i}},\mathbf{\hat{C}_{p_i}}\dots,\mathbf{\hat{C}_{p_i}}]^\mathbf{T}}
\end{aligned}
\label{CCPT1_Synthesis}
\end{equation}
and $\boldsymbol{\beta^{(1)}}$ is the transform coefficient vector.
Here $\mathbf{C}$ is a non-orthogonal matrix (refer \textbf{Theorem \ref{Th}}).
So,
\begin{equation}
{\boldsymbol{\beta^{(1)}}={\mathbf{C^{-1}}}\mathbf{x}}.
\label{CCPT1_Analysis} 
\end{equation}
This transformation from $\mathbf{x}$ to $\boldsymbol{\beta^{(1)}}$ is known as \textit{Complex Conjugate Periodic Transform (CCPT) \cite{Shah}, denoted as CCPT$^{(1)}$}. Both (\ref{CCPT1_Synthesis}) and (\ref{CCPT1_Analysis}) together form a CCPT$^{(1)}$ pair.
%
%

In a similar way, if we construct a ${p_i}{\times}{p_i}$ circulant matrix $\mathbf{G_{p_i,k}^{(2)}}$ using $c_{{p_i},k}^{(2)}(n)$, then $r(\mathbf{G_{{p_i},k}^{(2)})}=2$.
Further, $\mathbf{G_{{p_i},k}^{(2)}}$ can be factorized as (refer \textbf{Theorem} \ref{Th1} proof for $M=2$): 
\begin{equation}
\mathbf{G_{p_i,k}^{(2)}}=-j(\mathbf{F}\mathbf{\hat{F}^H}),\text{ where}
\label{Decomposition}
\end{equation} 
\begin{equation}
\nonumber
\footnotesize
\begin{aligned}
\mathbf{F^H}&=\begin{bmatrix}
S_{p_i,{p_i}-k}(0)&S_{p_i,{p_i}-k}(1)&\dots&S_{p_i,{p_i}-k}(p_i-1) \\
S_{p_i,k}(0)&S_{p_i,k}(1)&\dots&S_{p_i,k}({p_i}-1)
\end{bmatrix}_{2\times {p_i}}\ \\ 
\normalsize
\text{and}\\
\footnotesize
\mathbf{\hat{F}^H}&=\begin{bmatrix}
S_{p_i,{p_i}-k}(0)&S_{p_i,{p_i}-k}(1)&\dots&S_{p_i,{p_i}-k}(p_i-1)\\
-S_{p_i,k}(0)&-S_{p_i,k}(1)&\dots&-S_{p_i,k}({p_i}-1)
\end{bmatrix}_{2\times {p_i}}.
\end{aligned}
\normalsize
\end{equation}
From (\ref{Decomposition}), the column space of $\mathbf{G_{{p_i},k}^{(2)}}$ is same as the column space of $\mathbf{F}$ \cite{Strang}, which is $v_{{p_i},k}$. 
Moreover, one can check that \textit{the first two columns of $\mathbf{G_{{p_i},k}^{(2)}}$ {\em i.e.,} $c_{{p_i},k}^{(2)}$ and $c_{{p_i},k}^{(2)1}$ are linearly independent.
So, they act as a basis for CCS.}
This provides another alternative for $\mathbf{\hat{A}_{p_i}}$ as given below:
\begin{equation}
\mathbf{\hat{D}_{p_i}} = [\underbrace{{c}_{{p_i},k_1}^{(2)},\ {c}_{{p_i},k_1}^{(2)1}}_{\text{Basis of } v_{p_i,k_1}},\dots,\ \underbrace{{c}_{{p_i},k_{\frac{\varphi(p_i)}{2}}}^{(2)},\ {c}_{{p_i},k_{\frac{\varphi(p_i)}{2}}}^{(2)1}}_{\text{Basis of } v_{p_i,k_\frac{\varphi(p_i)}{2}}}]_{p_i{\times}\varphi(p_i)}.
\label{Rpi1}
\end{equation}
From \textbf{Theorem \ref{Th}} and (\ref{Rpi1}), there are a few points worth noting:
\begin{itemize}
\item As $\left<{c}_{{p_i},k_i}^{(2){l_1}},{c}_{{p_i},k_j}^{(2){l_2}}\right>= 0$, ${\forall\ {k_i{\neq}k_j}} $, the columns of $\mathbf{\hat{D}_{p_i}}$ are CCS wise orthogonal, this implies $r(\mathbf{\hat{D}_{p_i})} = \varphi(p_i)$.
\item Since $c_{{p_i},k_i}^{(2)}$ is a periodic sequence with period $p_i$, each column in $\mathbf{\hat{D}_{p_i}}$ is a $p_i$ periodic sequence. 
\item As $\left<c_{{p_i},k_i}^{(2){l_1}},c_{{p_j},k_j}^{(2){l_2}}\right>= 0$, $\forall {p_i}{\neq}{p_j}$, ${p_i|N}$, and ${p_j|N}$, the rank of 
$\mathbf{D}{\in}M_N(\mathbb{C})$ built using $\mathbf{\hat{D}_{p_i}}$ is equal to $N$. 
\end{itemize}
From the above three points $\mathbf{D}$ results in an NPM. Hence, any finite $N$ length signal $\mathbf{x}$ is represented/synthesized as:
\begin{equation}
\mathbf{x} = \mathbf{D}\boldsymbol{\beta^{(2)}} = {[\mathbf{D_{p_1}},\dots, \mathbf{D_{p_i}}, \dots, \mathbf{D_{p_m}}]}_{N{\times}N}{\boldsymbol\beta^{(2)}},
\label{Synthesis}
\end{equation}
where $\boldsymbol{\beta^{(2)}}$ is the transform coefficient vector.
Along with this, 
the product
$[\mathbf{D_{p_i}}]^\mathbf{T}[\mathbf{D_{p_j}}]=[\mathbf{0}]_{{\varphi(p_i)}\times{\varphi(p_j)}}$, ${\forall}\ {p_i}{\neq}{p_j}$, {\it i.e.,} $s_{p_i}$ is orthogonal to $s_{p_j}$. If the subspaces are orthogonal, they can be uniquely determined as \textit{Ramanujan Subspaces} \cite{7109930}.
From \textbf{Theorem \ref{Th}}, 
$\big<c_{p_i,k}^{(2)}(n),c_{p_i,k}^{(2)1}(n)\big>{\neq}0$, hence $\mathbf{D}$ is a non-orthogonal matrix.
So,
\begin{equation}
{\boldsymbol{\beta^{(2)}}={\mathbf{D^{-1}}}\mathbf{x}}.
\label{CCPT2_Analysis}
\end{equation}
The above transform is known as \textit{CCPT} and denoted as \textit{CCPT$^{(2)}$}. Both (\ref{Synthesis}) and (\ref{CCPT2_Analysis}) together form a CCPT$^{(2)}$ pair.
The non-orthogonality of the transformation matrices $\mathbf{C}$ and $\mathbf{D}$ results in both CCPT$^{(1)}$ and CCPT$^{(2)}$ as non-orthogonal transforms.
One can say that $\mathbf{C}$ and $\mathbf{D}$ are partially orthogonal as their columns are CCS wise orthogonal.
In the following section, we propose an orthogonal basis for CCS, this leads to the construction of another NPM $\mathbf{E}$.
\section{Orthogonal CCPT (OCCPT) and Its Properties}
Here, both $c_{{p_i},k}^{(1)}$ and $c_{{p_i},k}^{(2)}$ together act as a basis for $v_{p_i,k}$.
\begin{theorem} Given $v_{p_i,k}$ is the subspace spanned by $\{S_{p_i,k}(n),S_{p_i,p_i-k}(n)\}$ and $u_{p_i,k}$ is the subspace spanned by $\{c_{p_i,k}^{(1)}(n),c_{p_i,k}^{(2)}(n)\}$. Then $v_{p_i,k}$ is equal to $u_{p_i,k}$.
\label{Th4}
\end{theorem}
\begin{proof}Let $x_{p_i,k}{\in}v_{p_i,k}$, then there exist ${\alpha_1},{\alpha_2}{\in}\mathbb{C}$ s.t. $x_{p_i,k}(n)= {\alpha_1}S_{p_i,k}(n)+{\alpha_2}S_{p_i,p_i-k}(n)$. Using Euler's identity, we can rewrite $x_{p_i,k}(n)= {\gamma_1}c_{p_i,k}^{(1)}(n)+{\gamma_2}c_{p_i,k}^{(2)}(n)$, ${\gamma_1},{\gamma_2}{\in}\mathbb{C}$, hence $x_{p_i,k}{\in}u_{p_i,k}$. Similarly, any $y_{p_i,k}{\in}u_{p_i,k}$ also belongs to $v_{p_i,k}$. Hence $v_{p_i,k}$ is equal to $u_{p_i,k}$.\end{proof}
So, another alternative to $\mathbf{\hat{A}_{p_i}}$ can be written as given below:
\begin{equation}
\nonumber
\begin{aligned}
\mathbf{\hat{E}_{p_i}} = [\underbrace{{c}_{{p_i},k_1}^{(1)},\ {c}_{{p_i},k_1}^{(2)}}_{\text{Basis of } v_{p_i,k_1}},\dots,\ \underbrace{{c}_{{p_i},k_{\frac{\varphi(p_i)}{2}}}^{(1)},\ {c}_{{p_i},k_{\frac{\varphi(p_i)}{2}}}^{(2)}}_{\text{Basis of } v_{p_i,k_\frac{\varphi(p_i)}{2}}}]_{{p_i}{\times}\varphi(p_i)}.
\end{aligned}
\end{equation}
If $p_i = 1$ (or) $2$, then $\mathbf{\hat{E}_{p_i}}{\in}M_{p_i,1}$ and ${c}_{{p_i},k}^{(1)}={c}_{{p_i},k}^{(2)}$. So, we can consider either ${c}_{{p_i},k}^{(1)}$ or ${c}_{{p_i},k}^{(2)}$. 
From \textbf{Theorem} \ref{Th3}, we can summarize two points:
\begin{itemize}
\item As $\left<c_{{p_i},k_i}^{(1)},c_{{p_i},k_j}^{(2)}\right>= 0$, $\forall\ {p_i|N},\ {k_i} \text{ and } {k_j}{\in}\hat{U}_{p_i}$, the columns of $\mathbf{\hat{E}_{p_i}}$ are orthogonal to each other. 
\item As $\left<c_{{p_i},k_i}^{(1)},c_{{p_j},k_j}^{(2)}\right>= 0$, $\forall\ {p_i|N}$, ${p_j|N}$, ${k_i}{\in}\hat{U}_{p_i}$ and ${k_j}{\in}\hat{U}_{p_j}$, the columns of the matrix $\mathbf{E}$ constructed using $\mathbf{\hat{E}_{p_i}}$ are mutually orthogonal. 
\end{itemize}
Furthermore, $\mathbf{E}$ is an NPM as both CCPSs are periodic. So, an $N$ length signal $x(n)$ is represented/synthesized as
\begin{equation}
\begin{aligned}
&\quad \mathbf{x} = \mathbf{E}\boldsymbol{\beta} = {[\mathbf{E_{p_1}},\dots, \mathbf{E_{p_i}}, \dots, \mathbf{E_{p_m}}]}_{N{\times}N} \boldsymbol{\beta},
\normalsize
\end{aligned}
\label{Rpi4}
\end{equation}
where $\boldsymbol{\beta}$ is the transform coefficient vector.
This transform
is known as \textit{Orthogonal CCPT}.
Though the columns of $\mathbf{E}$ are mutually orthogonal, the product $\mathbf{E^{T}E} = 2NM\mathbf{I}$, where  $M =\begin{cases}
	\frac{1}{2},& \text{if}\ p_i=1\ \text{(or)}\ 2\\
    1, & \text{if }\ {p_i{\geq}3}
\end{cases}$, and $\mathbf{I}{\in}M_N(\mathbb{C})$ is an identity matrix.
So the resultant analysis equation is
\begin{equation}
\boldsymbol{\beta} = \frac{1}{2NM}\mathbf{E^T}\mathbf{x}.
\end{equation}
The representation in (\ref{Rpi4}) can be written algebraically as 
\begin{equation}
x(n) = \sum\limits_{p_i|N}x_{p_i}(n)= \sum_{{p_i}|N} \sum\limits_{\substack{{k}=1\\(k,p_i)=1}}^{\floor*{\frac{p_i}{2}}}\underbrace{{\beta_{0{k}i}c_{p_i,k}^{(1)}(n)}+{\beta_{1{k}i}c_{p_i,k}^{(2)}(n)}}_{x_{p_i,k}{\in}v_{p_i,k}},
\label{Ortho_CCPT_Synth}
\end{equation}
where $0{\leq}n{\leq}N-1$ and $x_{p_i}{\in}s_{p_i}$.
Manipulating \eqref{Ortho_CCPT_Synth} algebraically with $c_{p_j,k_1}^{(1)}(n)$ leads to (where $p_j|N$ and $k_1{\in}\hat{U}_{p_j}$),
\begin{equation}
\begin{aligned}
\sum\limits_{n=0}^{N-1}x(n)c_{p_j,k_1}^{(1)}(n)=\sum_{{p_i}|N}&\sum\limits_{\substack{{k}=1\\(k,p_i)=1}}^{\floor*{\frac{p_i}{2}}}\bigg({\mathbf{Q}}+{\mathbf{R}}\bigg),\text{ here}
\end{aligned}
\end{equation}
\begin{equation}
\nonumber
\footnotesize
\mathbf{Q}=\sum\limits_{n=0}^{N-1}{\beta_{0{k}i}c_{p_i,k}^{(1)}(n)c_{p_j,k_1}^{(1)}(n)}=\ \begin{cases}
	2N{M}{\beta_{0{k_1}j}},& \text{if}\ {p_i=p_j}\\& \text{and}\ {k=k_1}\\
    0, &\text{Otherwise}
\end{cases},
\normalsize
\end{equation}
and 
$\mathbf{R} =\sum\limits_{n=0}^{N-1}{\beta_{1{k}i}c_{p_i,k}^{(2)}(n)c_{p_j,k_1}^{(1)}(n)}$ $=0$.
Similarly, by manipulating with $c_{p_j,k_1}^{(2)}(n)$, 
we can obtain the following 
\begin{equation}
\begin{aligned}
&{\beta_{0{k}i}} = \frac{1}{2N{M}}\sum\limits_{n=0}^{N-1}x(n)c_{p_i,k}^{(1)}(n),\ {p_i}|N\ \&\ {k{\in}\hat{U}_{p_i}},\\
&{\beta_{1{k}i}} = \frac{1}{2NM}\sum\limits_{n=0}^{N-1}x(n)c_{p_i,k}^{(2)}(n),\ {p_i}|N\ \&\ {k{\in}\hat{U}_{p_i}}.
\label{Ortho_CCPT_Analy}
\end{aligned}
\end{equation}
Here \eqref{Ortho_CCPT_Analy} is an analysis equation and \eqref{Ortho_CCPT_Synth} is a synthesis equation, both together form an \textit{$N$-point OCCPT pair}.
\subsection{Properties} 
Let $x(n),\ x_1(n)$ and $x_2(n)$ be the signals, whose OCCPT coefficients are 
 $\Big(\beta_{0ki},\beta_{1ki}\Big)$, $\Big(\hat{\beta}_{0ki},\hat{\beta}_{1ki}\Big)$ and $\Big(\tilde{\beta}_{0ki},\tilde{\beta}_{1ki}\Big)$
 respectively. These relationships are denoted as
\begin{equation}
\nonumber
\begin{aligned}
x(n)&{\xleftrightarrow{\text{N - OCCPT}}}\Big(\beta_{0ki},\beta_{1ki}\Big),
x_1(n){\xleftrightarrow{\text{N - OCCPT}}}\Big(\hat{\beta}_{0ki},\hat{\beta}_{1ki}\Big)\\ 
&\qquad\text{and }x_2(n){\xleftrightarrow{\text{N - OCCPT}}} \Big(\tilde{\beta}_{0ki},\tilde{\beta}_{1ki}\Big).
\end{aligned}
\end{equation}
Then we can derive the following properties:

%
\subsubsection{Circular Shift of a Sequence}
The $N$-point OCCPT of $x\big(((n-m))_N\big)$, for an arbitrary $m{\in}\mathbb{Z}$, is defined as
\begin{equation}
x_1(n) = x\Big(((n-m))_N\Big)\ {\xleftrightarrow{\text{N - OCCPT}}}\ \Big(\hat{\beta}_{0ki},\hat{\beta}_{1ki}\Big).
\end{equation}
\begin{equation}
\text{If }p_i = 1\text{ (or) }2:\ \hat{\beta}_{0ki} = {\beta}_{0ki}cos\left(\theta\right).
\end{equation}
\begin{equation}
\text{If }p_i{\geq}3:\begin{bmatrix}
\hat{\beta}_{0ki}\\
\\
\hat{\beta}_{1ki}
\end{bmatrix}= \begin{bmatrix}
cos\left(\theta\right) & sin\left(\theta\right) \\
\\
-sin\left(\theta\right) & cos\left(\theta\right)
\end{bmatrix}\begin{bmatrix}
\beta_{0ki} \\
\\
\beta_{1ki}
\end{bmatrix},
\label{Translation}
\end{equation}
where $\theta = \frac{2{\pi}k{((-m))_N}}{p_i}$. So, the circular delay in time results in proportionate rotation of transform coefficients.

\subsubsection{Circular Convolution} 
The $N$-point OCCPT of $x(n) = x_1(n){\circledast}x_2(n)$ is defined as
\begin{equation}
\nonumber
x(n) = \sum\limits_{l=0}^{N-1}x_1(l)x_2\Big(((n-l))_N\Big) \ {\xleftrightarrow{\text{N - OCCPT}}}\ \Big(\beta_{0ki},\beta_{1ki}\Big).
\end{equation}
\begin{equation}
\text{If }p_i = 1\text{ (or) }2:\ {\beta}_{0ki} = N\ \hat{\beta}_{0ki}\ \tilde{\beta}_{0ki}.
\end{equation}
\begin{equation}
\text{If }p_i{\geq}3:\ \begin{bmatrix}
{\beta}_{0ki}\\
\\
{\beta}_{1ki}
\end{bmatrix}= N\begin{bmatrix}
\hat{\beta}_{0ki} & -\hat{\beta}_{1ki} \\
\\
\hat{\beta}_{1ki} & \hat{\beta}_{0ki}
\end{bmatrix}\begin{bmatrix}
\tilde{\beta}_{0ki} \\
\\
\tilde{\beta}_{1ki}
\end{bmatrix}.
\label{Conv1}
\end{equation}
From above, the transform coefficients of circular convolution are satisfying the commutative property.
\subsubsection{Parseval's Relation}  
The OCCPT conserves the energy of a given signal.
In specific,
\begin{equation}
\begin{aligned}
\sum\limits_{n=0}^{N-1}|x(n)|^2 &= N\Big(|\beta_{011}|^2+|\beta_{012}|^2\Big)\\&+2N\sum_{\substack{{{p_i}|N}\\{p_i{\geq}3}}} \sum\limits_{\substack{{k}=1\\(k,p_i)=1}}^{\floor*{\frac{p_i}{2}}}|\beta_{0ki}|^2+|\beta_{1ki}|^2.
\end{aligned}
\end{equation}
Here the term $\beta_{012} = 0$, if $2$ is not a divisor of $N$. 
Proofs for the above three properties are given in the appendix.
\subsubsection{Periodicity}
Using the periodicity property (with respect to $k$) of CCPSs, it can be proved that $\left(\beta_{0ki},\beta_{1ki}\right)$ are periodic with period $N$, i.e.,
\begin{equation}
\beta_{0(k+N)i} = \beta_{0ki}\ \text{and}\ \beta_{1(k+N)i} = \beta_{1ki}.
\end{equation}
\subsection{Relation Between Orthogonal CCPT and DFT Coefficients}
Let $x(n){\in}\mathbb{C}^N$, then the DFT coefficients $X(k)$, $0\leq k \leq {N-1}$ are obtained by performing 
the dot product between $x(n)$ and $e^{\frac{-j2{\pi}kn}{N}}=cos\Big(\frac{2{\pi}kn}{N}\Big)-jsin\Big(\frac{2{\pi}kn}{N}\Big)$. 
For a given $k$, let $k_i = \frac{k}{d_i}$ and $p_i = \frac{N}{d_i}$ where $d_i=(k,N)$. 
Then the basis (\textit{cosine} and \textit{sine} terms) of $v_{p_i,k_i}$ establish relationship between $X(k)$ and $\Big(\beta_{0{k_i}{i}},\beta_{1{k_i}{i}}\Big)$. Let
\begin{equation}
x(n)\ {\xleftrightarrow{\text{N - DFT}}}\ X(k)\ \&\ x(n)\ {\xleftrightarrow{\text{N - OCCPT}}}\ \Big(\beta_{0{k_i}i},\beta_{1{k_i}i}\Big),
\end{equation}
where $\beta_{0{k_i}i} = \hat{\beta}_{0{k_i}i}+j\tilde{\beta}_{0{k_i}i}$ and $\beta_{1{k_i}i} = \hat{\beta}_{1{k_i}i}+j\tilde{\beta}_{1{k_i}i}$.
Then, by exploiting the analysis equation of DFT, OCCPT and symmetry properties of CCPSs, we can establish the following relation for every $p_i{\in}D_N$ and $k_i{\in}\hat{U}_{p_i}$:\\
\begin{equation}
\text{If }p_i = 1\text{ (or) }2:X(k)= X\left(\frac{N{k_i}}{p_i}\right)= N\beta_{0{k_i}i}.
\label{Relation_with_DFT}
\end{equation}
\begin{equation}
\begin{aligned}
&\text{If }p_i{\geq}3:\ X(k)= X\left(\frac{N{k_i}}{p_i}\right)\\
&= \begin{cases}
	N\left[\Big(\hat{\beta}_{0{k_i}i}+\tilde{\beta}_{1{k_i}i}\Big)+j\Big(\tilde{\beta}_{0{k_i}i}-\hat{\beta}_{1{k_i}i}\Big)\right],\ \text{if}\ {k_i}{\in}\hat{U}_{p_i}\\ \\
    N\left[\Big(\hat{\beta}_{0{k_i}i}-\tilde{\beta}_{1{k_i}i}\Big)+j\Big(\tilde{\beta}_{0{k_i}i}+\hat{\beta}_{1{k_i}i}\Big)\right],\ \text{if}\ {k_i}{\in}\tilde{U}_{p_i}
\end{cases}.
\end{aligned}
\label{Relation_with_DFT1}
\end{equation}
If $x(n){\in}\mathbb{R}^N$, then $\tilde{\beta}_{0{k_i}i} = \tilde{\beta}_{1{k_i}i} = 0$ in the above equation.
Since the DFT coefficients are periodic with $N$, consider $k=N$ whenever $k=0$, as it is an invalid case for $v_{N,k}$. 
From the relations given in \eqref{Relation_with_DFT} and \eqref{Relation_with_DFT1}, we can get both magnitude and phase information of a discrete frequency $\left(\frac{2{\pi}k}{N}\right)$ from OCCPT coefficients. 
The presence of circular downshift terms (in both $\mathbf{C}$ and $\mathbf{D}$) allows us to extract magnitude, and ut not the phase information from $\boldsymbol{\beta^{(1)}}$ and $\boldsymbol{\beta^{(2)}}$.
 
The notion of the period is explicitly available from the $N$-point OCCPT pair given in \eqref{Ortho_CCPT_Synth} and \eqref{Ortho_CCPT_Analy}, whereas the notion of frequency is represented in a succinct way. To get explicit frequency information, a reinterpretation of OCCPT followed by a DIT based fast computation algorithm for it are discussed in the following section.
\section{Fast OCCPT (FOCCPT)}
\subsection{Reinterpretation of OCCPT}
Consider the following sets of irreducible rational numbers
\begin{equation}
\begin{aligned}
H_1 &= \left\{\frac{k}{p_i}{\mid}\ 0{\leq}k{\leq}\floor*{\frac{{p_i}}{2}},(k,p_i)=1, p_i|N\right\},\\
H_2 &= \left\{\frac{\hat{k}}{p_i}{\mid}\ \floor*{\frac{{p_i}}{2}}+1{\leq}\hat{k}{\leq}{p_i-1},(\hat{k},p_i)=1, p_i|N, {p_i}>2\right\}.
\end{aligned}
\label{reint1}
\end{equation}
Both $H_1$ and $H_2$ can be rewritten as sets of all rational elements as given below:
\begin{equation}
\begin{aligned}
&H_1 = \left\{\frac{K}{N}{\mid}\ 0{\leq}K{\leq}\floor*{\frac{{N}}{2}}\right\}\ \text{and}\\ 
&H_2 = \left\{\frac{K}{N}{\mid}\ {\floor*{\frac{N}{2}}+1}{\leq}K{\leq}{N-1}\right\}. 
\end{aligned}
\label{reint2}
\end{equation}
In an $N$-point OCCPT, the total columns of $\mathbf{E}$ can be divided into two sets as
\begin{equation}
\nonumber
\begin{aligned}
H_3 &= \left\{ 2Mcos\left(\frac{2{\pi}kn}{p_i}\right){\mid}\ k{\in}{\hat{u}_{p_i}},{p_i}|N\right\},\\
H_4 
&=\left\{-2sin\left(\frac{2{\pi}\hat{k}n}{p_i}\right){\mid}\  \hat{k}=({p_i}-k){\in}{\tilde{u}_{p_i}} ,{p_i}|N,{p_i}>2  \right\}.
\end{aligned}
\end{equation}
Using \eqref{reint1} and \eqref{reint2}, we can reinterpret $H_3$ and $H_4$ as follows:
\begin{equation}
\begin{aligned}
H_3 &= \left\{2\hat{M}cos\left(\frac{2{\pi}Kn}{N}\right){\mid}\ 0{\leq}K{\leq}\floor*{\frac{N}{2}}\right\},\\
H_4 &= \left\{-2sin\left(\frac{2{\pi}Kn}{N}\right){\mid}\ {\floor*{\frac{N}{2}}+1}{\leq}K{\leq}{N-1}\right\},
\end{aligned}
\end{equation}
where, if $N$ is even then
$ \hat{M} =\begin{cases}
	\frac{1}{2},& \text{if}\ K=0\ \text{(or)}\ \frac{N}{2}\\
    1, &\ \text{otherwise}
\end{cases}\normalsize$ and if $N$ is odd then 
$ \hat{M} =\begin{cases}
	\frac{1}{2},& \text{if}\ K=0\\
    1, &\ \text{otherwise}
\end{cases}\normalsize.$
That is, $\mathbf{E}$ can be rewritten with some permutations of its columns as given below:
\begin{equation}
\begin{aligned}
\mathbf{\hat{E}}& = \Bigg[2\hat{M}cos\left(\frac{2{\pi}(0)}{N}n\right),\dots,2\hat{M}cos\left(\frac{2{\pi}\left(\floor*{\frac{N}{2}}\right)}{N}n\right),\\
&-2sin\left(\frac{2{\pi}{\left(\floor*{\frac{N}{2}}+1\right)}}{N}n\right),\dots,-2sin\left(\frac{2{\pi}(N-1)}{N}n\right)\Bigg].
\end{aligned}
\end{equation}
\normalsize
Now, 
we can rewrite the $N$-point OCCPT pair as
\begin{equation}
\beta(K) = \begin{cases}
	\frac{1}{N}\sum\limits_{n=0}^{N-1}x(n)cos\left(\frac{2{\pi}Kn}{N}\right),\ {0{\leq}K{\leq}\floor*{\frac{N}{2}}}\\
    -\frac{1}{N}\sum\limits_{n=0}^{N-1}x(n)sin\left(\frac{2{\pi}Kn}{N}\right),\ {\floor*{\frac{N}{2}}+1}{\leq}K{\leq}N-1.
\end{cases}
\label{Ortho_CCPT_Analy1}
\end{equation}
\footnotesize
\begin{equation}
x(n) = 2\sum\limits_{K=0}^{\floor*{\frac{N}{2}}}\hat{M}\beta(K)cos\left(\frac{2{\pi}Kn}{N}\right)-2\sum\limits_{K=\floor*{\frac{N}{2}}+1}^{N-1}\beta(K)sin\left(\frac{2{\pi}Kn}{N}\right).
\end{equation}
\normalsize
The relation between the coefficients given in \eqref{Ortho_CCPT_Analy} and \eqref{Ortho_CCPT_Analy1} is 
\begin{equation}
\beta_{0ki} = \beta\left(\frac{Nk}{p_i}\right)\ \text{and}\ \beta_{1ki} = \beta\left(\frac{N\hat{k}}{p_i}\right).
\end{equation}
\subsection{Decimation-In-Time FOCCPT (DIT-FOCCPT) Algorithm}
Here,
an $N$ length sequence $x(n)$ is decomposed into successively smaller sub-sequences.
The $N$-point OCCPT of $x(n)$ is computed by combining the OCCPT of these sub-sequences. 
The symmetry properties of CCPSs are exploited in the combining process.
As an initial step, we consider $N=2^v$, $v{\in}\mathbb{N}$ 
(similar to radix-$2$ DIT-FFT \cite{oppenheim1975digital}), 
this allows us to decompose $x(n)$ into two $\frac{N}{2}$ length sequences $h(n)$ and $g(n)$, 
where 
$h(n) = x(2n)$ and $g(n) = x(2n+1)$. 
By using the odd symmetry of $sin(.)$ function, the analysis equation of OCCPT given in \eqref{Ortho_CCPT_Analy1} is modified as
\begin{equation}
\footnotesize
\begin{aligned}
&N\beta(K) = X_x(K) = 
	\sum\limits_{n=0}^{N-1}x(n)cos\left(\frac{2{\pi}Kn}{N}\right),\ {0{\leq}K{\leq}\frac{N}{2}}\\
&N\beta(N-K) = Y_x(K) = \sum\limits_{n=0}^{N-1}x(n)sin\left(\frac{2{\pi}Kn}{N}\right),\ 1{\leq}K{\leq}\frac{N}{2}-1
\end{aligned}.
\normalsize
\label{Ortho_CCPT_Analy2}
\end{equation}
Here $X_x(K+N) = X_x(K) \text{ and } Y_x(K+N) = Y_x(K).$
Moreover, 
$X_x(N-K) = X_x(K)\text{ and }Y_x(N-K) = -Y_x(K)$ over the range of $0$ to $N$.
Now by decomposing $x(n)$ into $h(n)$ and $g(n)$, we obtain
\begin{equation}
\begin{aligned}
X_x(K)& = X_h(K)+ cos\left(\frac{2{\pi}K}{N}\right)X_g(K)
\\&-sin\left(\frac{2{\pi}K}{N}\right)Y_g(K) = N\beta(K),\ 
{0{\leq}K{\leq}\frac{N}{4}}\normalsize{\text{ and}} 
\end{aligned}
\label{mixingeqn1}
\normalsize
\end{equation}
\begin{equation}
\begin{aligned}
Y_x(K)& = Y_h(K)+ cos\left(\frac{2{\pi}K}{N}\right)Y_g(K)\\
&+sin\left(\frac{2{\pi}K}{N}\right)X_g(K) = N\beta(N-K),\ 
{1{\leq}K{\leq}\frac{N}{4}}.
\end{aligned}
\label{mixingeqn2}
\normalsize
\end{equation}
\begin{equation}
\begin{aligned}
\normalsize{\text{Where, }}&X_f(K) = \sum\limits_{n=0}^{\frac{N}{2}-1}f(n)cos\left(\frac{2{
\pi}Kn}{\frac{N}{2}}\right),\ 0{\leq}K{\leq}\frac{N}{4}\\
\normalsize{\text{and }}
&Y_f(K) = \sum\limits_{n=0}^{\frac{N}{2}-1}f(n)sin\left(\frac{2{
\pi}Kn}{\frac{N}{2}}\right),\ 1{\leq}K{\leq}\frac{N}{4}-1\\
\end{aligned}
\end{equation}
\normalsize
represents the $\frac{N}{2}$-point OCCPT of $f(n)$, here $f(n)$ can be either $h(n)$ or $g(n)$.
Here, 
$X_x(K)$ is computed for $0{\leq}K{\leq}\frac{N}{4}$, since the range of $X_h(K)$ and $X_g(K)$ is $0{\leq}K{\leq}\frac{N}{4}$.
The remaining $\frac{N}{4}$ coefficients of $X_x(K)$ are computed using the symmetry property of $X_h(K)$, $X_g(K)$ and $Y_g(K)$, i.e.,
\begin{equation}
\begin{aligned}
&X_x\left(\frac{N}{2}-K\right) = X_h(K)- cos\left(\frac{2{\pi}K}{N}\right)X_g(K)
\\&+sin\left(\frac{2{\pi}K}{N}\right)Y_g(K) = N\beta\left(\frac{N}{2}-K\right),\ 
{0{\leq}K{\leq}\frac{N}{4}-1}
\end{aligned}.
\label{mixingeqn3}
\end{equation}
\normalsize
Similarly, 
the remaining coefficients of $Y_x(K)$ are computed using the symmetry property of $Y_h(K)$, $Y_g(K)$ and $X_g(K)$, i.e.,
\begin{equation}
\begin{aligned}
Y_x&\left(\frac{N}{2}-K\right) = -Y_h(K)+ cos\left(\frac{2{\pi}K}{N}\right)Y_g(K)\\
&+sin\left(\frac{2{\pi}K}{N}\right)X_g(K) = N\beta\left(\frac{N}{2}+K\right),\ 
{1{\leq}K{\leq}\frac{N}{4}-1}
\end{aligned}.
\label{mixingeqn4}
\end{equation}
\normalsize
Note that in \eqref{mixingeqn1}, \eqref{mixingeqn2} and \eqref{mixingeqn3} we have the terms $Y_h(K)$ and $Y_g(K)$, with possible $K$ values to be $0$ (or) $\frac{N}{4}$, but the actual range of $Y_h(K)$ and $Y_g(K)$ is $1{\leq}K{\leq}\frac{N}{4}-1$. Moreover, $Y_h(K) = Y_g(K) = 0$ whenever $K=0$ (or) $\frac{N}{4}$. 
So, the equations \eqref{mixingeqn1}, \eqref{mixingeqn2} and \eqref{mixingeqn3} can be further reduced as follows:
\begin{equation}
X_x(K) = \begin{cases}
\begin{aligned}
&X_h(K)+cos\left(\frac{2{\pi}K}{N}\right)X_g(K),\ \text{if }K=0\text{ (or) }\frac{N}{4}\\
	&X_h(K)+cos\left(\frac{2{\pi}K}{N}\right)X_g(K)\\&{\qquad}-sin\left(\frac{2{\pi}K}{N}\right)Y_g(K),\ 
	\text{if }1{\leq}K{\leq}\frac{N}{4}-1
	\end{aligned}
\end{cases}
\label{mixingeqn5}
\end{equation}
\normalsize 

\begin{equation}
X_x\left(\frac{N}{2}-K\right) = \begin{cases}
\begin{aligned}
&X_h(K)-cos\left(\frac{2{\pi}K}{N}\right)X_g(K),\ \text{if }K=0\\
	&X_h(K)-cos\left(\frac{2{\pi}K}{N}\right)X_g(K)
	\\&\ +sin\left(\frac{2{\pi}K}{N}\right)Y_g(K),\ 
	\text{if }1{\leq}K{\leq}\frac{N}{4}-1
	\end{aligned}
\end{cases}
\label{mixingeqn6}
\end{equation}
\normalsize

and
\begin{equation}
Y_x(K) = \begin{cases}
\begin{aligned}
Y_h(K)&+cos\left(\frac{2{\pi}K}{N}\right)Y_g(K)\\
&+sin\left(\frac{2{\pi}K}{N}\right)X_g(K),\ \text{if }1{\leq}K{\leq}\frac{N}{4}-1\\
sin&\left(\frac{2{\pi}K}{N}\right)X_g(K),\quad \text{if }K = \frac{N}{4}
\end{aligned}
\end{cases}
\label{mixingeqn8}
\end{equation}
\normalsize
The equations \eqref{mixingeqn4}-\eqref{mixingeqn8} correspond to the decomposition of original $N$-point OCCPT into two $\frac{N}{2}$-point OCCPT computations. 
As $N=2^v$, we can further decompose each $\frac{N}{2}$-point OCCPT into two $\frac{N}{4}$-point OCCPTs. 
This process is repeated for $v = log_2(N)$ times.
Fig. \ref{The flow graph of complete DIT decomposition of an $8$-point OCCPT.} depicts the complete DIT decomposition flow graph of an $8$-point OCCPT computation.
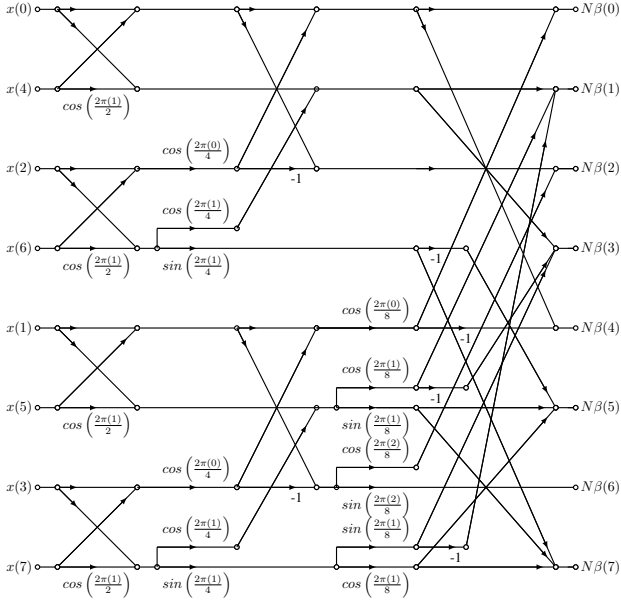
\begin{figure}
\begin{adjustbox}{max width=\textwidth}
\renewcommand{\arraystretch}{1.5}
\scalebox{0.53}{
\tikzset{%
    block/.style    = {draw, thick, rectangle, minimum height = 3em,
        minimum width = 3em},
    gain/.style     = {draw, thick, isosceles triangle, minimum height = 3em,
        isosceles triangle apex angle=60},
    port/.style     = {inner sep=0pt, font=\tiny},
    sum/.style n args = {4}{draw, circle, node distance = 2cm, minimum size=6mm, alias=sum,
        append after command={
            node at (sum.north) [port, below=6pt] {$#1$}
            node at (sum.west) [port, right=6pt] {$#2$}
            node at (sum.south) [port, above=6pt] {$#3$}
            node at (sum.east) [port, left=6pt] {$#4$}
        },
    }, 
    joint/.style    = {circle, draw, fill, inner sep=0pt, minimum size=2pt},
    input/.style    = {coordinate}, 
    output/.style   = {coordinate} 
}
\newcommand{\suma}{\Large$+$}
\newcommand{\inte}{$\displaystyle \int$}
\newcommand{\derv}{\huge$\frac{d}{dt}$}
\begin{circuitikz} [color=black, thick][size=0.125]
\draw (-1,2) to [short,o-o] (-.5,2);
    \draw[-latex] (-.5,2) -- (.5,2) node[below]{$cos\left(\frac{2\pi(1)}{2}\right)$};
    \draw (-.5,2) to [short,o-o] (1.5,2) node[right=4pt, above=1pt]{};
    \draw[-latex] (-.5,2) -- (1.25,3.75) node[below] {};   
     \draw[-latex] (-.5,4) -- (0,3.5) node[below] {}; 
     \draw[-latex] (-.5,4) -- (0,4) node[below] {};

\draw (-1,4) to [short,o-o] (-.5,4);
   \draw (-.5,4) to [short,o-o] (1.5,4) node[right=4pt, above=1pt]{};

    \draw (-.5,2) to [short,o-o] (1.5,4);
    
    \draw (-.5,4) to [short,o-o] (1.5,2);

\draw (-1,6) to [short,o-o] (-.5,6) ;
    \draw[-latex] (-.5,6) -- (.5,6) node[below] {$cos\left(\frac{2\pi(1)}{2}\right)$};
    \draw (-.5,6) to [short,o-o] (1.5,6) node[right=4pt, above=1pt]{};
  
    \draw (-.5,8) to [short,o-o] (1.5,8) node[right =4pt, above=1pt]{};
\draw (-1,8) to [short,o-o] (-.5,8);

    \draw (-.5,6) to [short,o-o] (1.5,8); 
       
    \draw (-.5,8) to [short,o-o] (1.5,6);
     \draw[-latex] (-.5,6) -- (1.25,7.75) node[below] {};   
     \draw[-latex] (-.5,8) -- (0,7.5) node[below] {}; 
     \draw[-latex] (-.5,8) -- (0,8) node[below] {};

\draw (-1,0) to [short,o-o] (-.5,0);

    \draw (-.5,0) to [short,o-o] (1.5,0) node[right=4pt, above=1pt]{};

\draw (-1,-2) to [short,o-o] (-.5,-2);
    \draw[-latex] (-.5,-2) -- (.5,-2) node[below]{$cos\left(\frac{2\pi(1)}{2}\right)$};
    \draw (-.5,-2) to [short,o-o] (1.5,-2) node[right=4pt, above=1pt]{};
    \draw[-latex] (-.5,-2) -- (1.25,-0.25) node[below] {};   
     \draw[-latex] (-.5,0) -- (0,-.5) node[below] {}; 
     \draw[-latex] (-.5,0) -- (0,0) node[below] {};

    \draw (-.5,0) to [short,o-o] (1.5,-2);
   
    \draw (-.5,-2) to [short,o-o] (1.5,0);

\draw (-1,-4) to [short,o-o] (-.5,-4);
   \draw (-.5,-4) to [short,o-o] (1.5,-4) node[right=4pt, above=1pt]{};

    \draw[-latex] (-.5,-6) -- (.5,-6) node[below] {\textbf{$cos\left(\frac{2\pi(1)}{2}\right)$}};
    \draw (-.5,-6) to [short,o-o] (1.5,-6) node[right=4pt, above=1pt]{};
\draw (-1,-6) to [short,o-o] (-.5,-6) ;
    \draw[-latex] (-.5,-6) -- (1.25,-4.25) node[below] {};   
     \draw[-latex] (-.5,-4) -- (0,-4.5) node[below] {}; 
     \draw[-latex] (-.5,-4) -- (0,-4) node[below] {};

    \draw (-.5,-4) to [short,o-o] (1.5,-6);
    
    \draw (-.5,-6) to [short,o-o] (1.5,-4);
\draw (1.5,-6) to [short,o-o] (2,-6);
    \draw[-latex] (2,-6) -- (3,-6) node[below] {\textbf{$sin\left(\frac{2{\pi}(1)}{4}\right)$}};
    \draw (2,-6) to [short,o-o] (6.5,-6) node[right=8pt, above=1pt]{};

\draw (2,-6) -- (2,-5.5) node[] {};
    \draw[-latex] (2,-5.5) -- (3,-5.5) node[above] {$cos\left(\frac{2{\pi}(1)}{4}\right)$};
    \draw (2,-5.5) to [short,-o] (4,-5.5);

    \draw[-latex] (4,-4) -- (5.5,-4) node[below] {-1};
    \draw (4,-4) to [short,o-o] (6.5,-4) node[right=8pt, above=1pt]{};
\draw (1.5,-4) to [short,o-o] (4,-4);
\draw [-latex] (1.5,-4) --(3,-4) node[above]{$cos\left(\frac{2{\pi}(0)}{4}\right)$}; 

    \draw (4,-5.5) to [short,o-o] (6,-2);
    
    \draw (4,-4) to [short,o-o] (6,0);

\draw (1.5,-2) to [short,o-o] (6.5,-2) node[right=8pt, above=1pt]{};

    \draw (1.5,0) to [short,o-o] (6,0) node[right=8pt, above=1pt]{};

    \draw (4,0) to [short,o-o] (6,-4);  
\draw (1.5,2) to [short,o-o] (2,2);
    \draw[-latex] (2,2) -- (3,2) node[below] {$sin\left(\frac{2{\pi}(1)}{4}\right)$};
    \draw (2,2) to [short,o-o] (8.5,2) node[right=8pt, above=1pt]{};

\draw (2,2) -- (2,2.5) node[] {};
    \draw[-latex] (2,2.5) -- (3,2.5) node[above] {$cos\left(\frac{2{\pi}(1)}{4}\right)$};
    \draw (2,2.5) to [short,-o] (4,2.5);

    \draw[-latex] (4,4) -- (5.5,4) node[below] {-1};
     \draw (4,4) to [short,-o] (6,4);
    \draw (6,4) to [short,o-o] (12.5,4) node[right=8pt, above=1pt]{};
\draw (1.5,4) to [short,o-o] (4,4);
\draw [-latex] (1.5,4) --(3,4) node[above]{$cos\left(\frac{2{\pi}(0)}{4}\right)$}; 

    \draw (4,2.5) to [short,o-o] (6,6) node[right=8pt, above=1pt]{};
    
    \draw (4,4) to [short,o-o] (6,8) node[right=8pt, above=1pt]{};
    \draw [-latex] (4,4) --(5.75,7.5) node[above]{}; 
    \draw [-latex] (4,2.5) --(5.75,5.5625) node[above]{}; 
    \draw[-latex] (4,8) -- (4.25,7.5) node[below] {};

     \draw[-latex] (4,8) -- (4.5,8) node[below] {};
     \draw[-latex] (4,0) -- (4.5,0) node[below] {};
     \draw [-latex] (4,-5.5) --(5.75,-2.4375) node[above]{};
     \draw[-latex] (4,-4) -- (5.75,-.5) node[below] {}; 
     \draw[-latex] (4,0) -- (4.25,-.5) node[below] {};

\draw (1.5,6) to [short,o-o] (8.5,6);

\draw (1.5,8) to [short,o-o] (4,8);
\draw (4,8) to [short,o-o] (6,8);
\draw (6,8) to [short,o-o] (8.5,8);
    \draw (4,8) to [short,o-o] (6,4);

    \draw[-latex] (8.5,2) -- (9,2) node[below] {-1};
    \draw (8.5,2) to [short,o-o] (9.75,2);
        \draw[-latex] (9.75,2) -- (11.8,-1.6444) node[below] {};
         \draw (9.75,2) to [short,o-o] (12,-2); 
                
       \draw (12,4) to [short,o-o] (12.5,4);

    \draw[-latex] (8.5,2) -- (11.8,-5.5428) node[] {};
     \draw (8.5,2) to [short,o-o] (12,-6);

   \draw[-latex] (8.5,6) -- (11.7,6) node[below] {};  

    \draw (8.5,8) to [short,o-o] (12,8);
\draw[-latex] (8.5,8) -- (9,8) node[below] {}; 
\draw[-latex] (8.5,4) -- (9,4) node[below] {}; 
\draw[-latex] (8.5,0) -- (11.75,7.4286) node[below] {}; 
\draw[-latex] (8.5,-3.5) -- (11.75,3.4642) node[below] {};

    \draw (12,8) to [short,o-o] (12.5,8);
  
    \draw[-latex] (8.5,6) -- (11.8,2.2285) node[] {};  
    \draw (8.5,6) to [short,o-o] (12,2);
      
    \draw (8.5,8) to [short,o-o] (12,0);

\draw (6,0) to [short,o-o] (8.5,0);
\draw[-latex] (6,0) -- (7.5,0) node[above] {$cos\left(\frac{2{\pi}(0)}{8}\right)$};
    \draw[-latex] (6,0) -- (9.75,0) node[below] {-1};
    \draw (8.5,0) to [short,o-o] (12,0);
    \draw (12,0) to [short,o-o] (12.5,0);
     \draw (8.5,6) to [short,o-o] (12,6);
     \draw (8.5,-2) to [short,o-o] (12,-2);
     \draw (8.5,-6) to [short,o-o] (12,-6);
	 \draw (12,6) to [short,o-o] (12.5,6);
	 \draw (12,2) to [short,o-o] (12.5,2);
     \draw (12,-2) to [short,o-o] (12.5,-2);
	 \draw (12,-6) to [short,o-o] (12.5,-6);

\draw[-latex] (8.5,8) -- (8.75,7.4285) node[below] {};

\draw (6.5,-2) to [short,o-o] (8.5,-2);

\draw [-latex] (7.5,-2) --(11.75,-2) node[above]{};      

\draw [-latex] (8.5,-2) --(11.8,-5.7714) node[above]{}; 
 \draw (8.5,-2) to [short,o-o] (12,-6);    

 \draw (8.5,-1.5) to [short,o-o] (9.75,-1.5);
     \draw[-latex] (8.5,-1.5) -- (9,-1.5) node[below] {-1};
  
 \draw (8.5,-5.5) to [short,o-o] (9.75,-5.5); 
     \draw[-latex] (8.5,-5.5) -- (9.5,-5.5) node[below] {-1};
     
      \draw[-latex] (9.75,-5.5) -- (11.75,4.7222) node[below] {};
		\draw (9.75,-5.5) to [short,o-o] (12,6);      
      
      \draw[-latex] (9.75,-1.5) -- (11.8,1.6888) node[below] {};
       \draw (9.75,-1.5) to [short,o-o] (12,2); 
      
		\draw[-latex] (8.5,-5.5) -- (11.75,1.4642) node[] {};
 		\draw (8.5,-5.5) to [short,o-o] (12,2); 		
		
    \draw (6.5,-2) -- (6.5,-1.5) node[] {};
    \draw[-latex] (6.5,-1.5) -- (7.5,-1.5) node[above] {$cos\left(\frac{2{\pi}(1)}{8}\right)$};
    \draw (6.5,-1.5) to [short,-o] (8.5,-1.5);
    
      \draw[-latex] (6.5,-2) -- (7.5,-2) node[below] {$sin\left(\frac{2{\pi}(1)}{8}\right)$};
    \draw (6.5,-4) -- (6.5,-3.5) node[] {};
    \draw[-latex] (6.5,-3.5) -- (7.5,-3.5) node[above] {$cos\left(\frac{2{\pi}(2)}{8}\right)$};
    \draw (6.5,-3.5) to [short,-o] (8.5,-3.5);
    
      \draw[-latex] (6.5,-4) -- (7.5,-4) node[below] {$sin\left(\frac{2{\pi}(2)}{8}\right)$};
    
    \draw (6.5,-6) -- (6.5,-5.5) node[] {};
    \draw[-latex] (6.5,-5.5) -- (7.5,-5.5) node[above] {$sin\left(\frac{2{\pi}(1)}{8}\right)$};
    \draw (6.5,-5.5) to [short,-o] (8.5,-5.5);
    
      \draw[-latex] (6.5,-6) -- (7.5,-6) node[below] {$cos\left(\frac{2{\pi}(1)}{8}\right)$};

    \draw (8.5,0) to [short,o-o] (12,8);

\draw[-latex] (8.5,-1.5) -- (11.5,4.9285) node[] {};
\draw (8.5,-1.5) to [short,o-o] (12,6);

    \draw (6,-4) to [short,o-o] (12.5,-4);
       
    \draw[-latex] (8.5,-6) -- (11.7,-6) node[below] {};
    
\draw (6.5,-6) to [short,o-o] (8.5,-6);

    \draw (8.5,-3.5) to [short,o-o] (12,4);
    
  \draw[-latex] (8.5,-6) -- (11.8,-2.2285) node[] {};
   \draw (8.5,-6) to [short,o-o] (12,-2); 
  
\draw (12.3,-2) to [short,-o] (12.5,-2);    
\draw (12.3,-6) to [short,-o] (12.5,-6);    
\draw (12.3,2) to [short,-o] (12.5,2);   
 \draw (12.3,6) to [short,-o] (12.5,6);

\draw (-1,2) node [left]{\textbf{$x(6)$}};   
\draw (-1,4) node [left]{\textbf{$x(2)$}};  
\draw (-1,6) node [left]{\textbf{$x(4)$}};     
\draw (-1,8) node [left]{\textbf{$x(0)$}};
\draw (-1,0) node [left]{\textbf{$x(1)$}};         
\draw (-1,-2) node [left]{\textbf{$x(5)$}};  
\draw (-1,-4) node [left]{\textbf{$x(3)$}};  
\draw (-1,-6) node [left]{\textbf{$x(7)$}}; 
\draw (12.5,8) node [right]{\textbf{$N\beta(0)$}};  
\draw (12.5,6) node [right]{\textbf{$N\beta(1)$}};  
\draw (12.5,4) node [right]{\textbf{$N\beta(2)$}};
\draw (12.5,2) node [right]{\textbf{$N\beta(3)$}};  
\draw (12.5,0) node [right]{\textbf{$N\beta(4)$}};  
\draw (12.5,-2) node [right]{\textbf{$N\beta(5)$}}; 
\draw (12.5,-4) node [right]{\textbf{$N\beta(6)$}};  
\draw (12.5,-6) node [right]{\textbf{$N\beta(7)$}}; 
\end{circuitikz}}
\end{adjustbox}
\caption{Flow graph of complete DIT decomposition of an $8$-point OCCPT computation.}
\label{The flow graph of complete DIT decomposition of an $8$-point OCCPT.}
\end{figure}
\subsection{Computational Complexity of N-Point DIT-FOCCPT}
Initially, we count the number of multiplications and additions required for $v^{th}$ stage (final stage), i.e., computing the $N$-point OCCPT from two $\frac{N}{2}$-point OCCPTs using \eqref{mixingeqn4}-\eqref{mixingeqn8}. 
Notice that in both \eqref{mixingeqn5} and \eqref{mixingeqn6}, we have the same multiplication terms with a difference in addition/subtraction operations. 
So, to find the number of multiplications, it is sufficient to consider one equation with a maximum variation range of $K$. 
Hence we consider \eqref{mixingeqn5} and it requires $2\left(\frac{N}{4}-1\right)+2 = \frac{N}{2}$ multiplications.
Whereas for additions we have to consider both the equations. They require $4\left(\frac{N}{4}-1\right)+3 = N-1$ additions.
Similarly, both  \eqref{mixingeqn4} and \eqref{mixingeqn8} require $2\left(\frac{N}{4}-1\right)+1 = \frac{N}{2}-1$ multiplications and $4\left(\frac{N}{4}-1\right) = N-4$ additions. 
As a result, we have the following for $v^{th}$ stage:

$\qquad\qquad\quad M_v = N-1,\ A_v = 2N-5,$
\newline
where $M_i$ and $A_i$ denote the number of multiplications and additions required in $i^{th}$ stage respectively.
Likewise, there are two $\frac{N}{2}$-point OCCPTs  computation from four $\frac{N}{4}$-point OCCPTs in $(v-1)^{th}$ stage. Hence

$\qquad\quad M_{v-1} = 2\left(\frac{N}{2}-1\right),\
A_{v-1}= 2\left(2\left(\frac{N}{2}\right)-5\right).$
\newline
Proceeding further, the first stage requires computation of $2$-point OCCPT for $\frac{N}{2}$ times. 
Here each $2$-point OCCPT requires one multiplication and two additions.
Therefore

$\qquad\qquad M_2 = \frac{N}{2}(2-1) = \frac{N}{2},\ A_2 = \frac{N}{2}(2) = N.$
\newline
Now, by combining each stage complexity
 we can count the total number of multiplications ($M_{total}$) and additions ($A_{total}$) required for $N$-point OCCPT, i.e.,
\begin{equation}
\nonumber
\footnotesize
\begin{aligned}
M_{total} &= 1(N-1)+2\left(\frac{N}{2}-1\right)+\dots+\frac{N}{4}(4-1)+\frac{N}{2}(2-1)\\
& = vN-(N-1) = Nlog_2(N)-N+1,\\
A_{total} &= 1(2N-5)+2\left(2\left(\frac{N}{2}\right)-5\right)+\dots+\frac{N}{4}\left(2(4)-5\right)+N\\
& = 2vN-5\left(\frac{N}{2}\right)+5-N = 2Nlog_2(N)-7\left(\frac{N}{2}\right)+5.
\end{aligned}
\normalsize
\end{equation}
Therefore, a given $x(n){\in}\mathbb{R}^N$ requires $Nlog_2(N)-N+1$ real multiplications and $2Nlog_2(N)-7\left(\frac{N}{2}\right)+5$ real additions for computing $N$-point OCCPT using DIT-FOCCPT. If $x(n){\in}\mathbb{C}^N$, then it requires $2Nlog_2(N)-2N+2$ real multiplications and $4Nlog_2(N)-7N+10$ real additions, as OCCPT is a linear transform.
\subsection{Comparison of Computational Complexity Between Different Transforms}
For a given $x(n){\in}\mathbb{C}^N$, the $N$-point OCCPT and DFT are computed using FOCCPT and FFT respectively, when $N=2^v,\ v{\in}\mathbb{N}$.
If $N{\neq}2^v$, both are computed 
using the direct method, even though there exist fast computation algorithms for DFT \cite{Burrus,Singleton}.
Whereas CCPT$^{(1)}$ and CCPT$^{(2)}$ are computed using direct method for both $N = 2^v$ and $N{\neq}2^v$ cases due to the non-orthogonality of transformation matrices $\mathbf{C}$ and $\mathbf{D}$. 
Even for RPT, we use the direct method for both cases. But if $N=2^v$, the RPT matrix is a sparse orthogonal matrix. To the best of our knowledge, there is no fast computation algorithm for computing RPT in the literature.

If $x{\in}\mathbb{C}^N$ and $y{\in}\mathbb{C}^N$, then the number of real multiplications and additions required for computing $\left<x,y\right>$ is $4N$ and $4N-2$ respectively.
Similarly, if $x{\in}\mathbb{C}^N$ and $y{\in}\mathbb{R}^N$, then it requires $2N$ and $2(N-1)$ real multiplications and additions respectively.
Using this, the computational complexity (\textit{in terms of the number of real multiplications and additions}) required for different transforms for a given $x(n){\in}\mathbb{C}^N$ is tabulated in Table \ref{tab:Comparison of computational complexity between different transformation techniques}. There are a few points we can summarize from Table \ref{tab:Comparison of computational complexity between different transformation techniques}:
\begin{itemize}
\item If $N=2^v$, then FOCCPT results in a
reduction of $2N-2$ real multiplications with an increase of $Nlog_2(N)-7N+10$ real additions over FFT.  Here the addition complexity becomes significant for a large value of $N$ ($N>2^8$).
\item If $N=2^v$ and $x(n){\in}\mathbb{R}^N$,
then FFT requires $2Nlog_2(N)$ real multiplications and $3Nlog_2(N)$ real additions. 
This implies FOCCPT has approximately $50\%$ reduction in computational complexity over FFT.
\item  The number of multiplications/additions required for the proposed transforms and RPT is approximately $50\%$ less when compared with DFT using the direct method.
\end{itemize}

Note that, both CCPT$^{(1)}$ and CCPT$^{(2)}$ require some additional complexity to find $\mathbf{C}^{-1}$ and $\mathbf{D}^{-1}$ respectively, along with the complexity given in Table \ref{tab:Comparison of computational complexity between different transformation techniques}.
This additional complexity is required even for RPT when $N{\neq}2^v$.
In this paper, both CCPT$^{(1)}$ and CCPT$^{(2)}$ are studied in brief.
A complete study of these transforms, especially the importance of circular shift operation in the matrices is one of our future works.
The proposed orthogonal and non-orthogonal transforms (CCPTs) may find applications in communication, image processing, control applications, etc., \cite{7263349,7842433,503278,8626481,8544037}. 
\begin{table}[h]
\centering
\centering
\footnotesize
\caption{C\scriptsize{OMPARISON OF COMPUTATIONAL COMPLEXITY BETWEEN DIFFERENT TRANSFORMATION TECHNIQUES}}
\label{tab:Comparison of computational complexity between different transformation techniques}
\begin{adjustbox}{max width=\textwidth}
\renewcommand{\arraystretch}{1.5}
\scalebox{0.7}{
\begin{tabular}{|M{1.6cm}|M{1.9cm}|M{2cm}|M{2.1cm}|M{2cm}|} \hline
\bfseries{\small{}} & 
\multicolumn{2}{|c|}{\bfseries{\footnotesize{If $N{\neq}2^{v}$}}} &  
\multicolumn{2}{|c|}{\bfseries{\footnotesize{If $N=2^{v}$}}} \\ \cline{2-5}
 & \bfseries{\footnotesize{\makecell{Number of Real\\ Multiplications}}} & \bfseries{\footnotesize{\makecell{Number of\\Real Additions}}} & \bfseries{\footnotesize{\makecell{Number of Real \\ Multiplications}}} & \bfseries{\footnotesize{\makecell{Number of\\Real Additions}}} \\ \hline 
\small{CCPT$^{(1)}$} &	$2{N^2}$	&	$2{N^2}-2N$ &	$2{N^2}$	&	$2{N^2}-2N$	\\	\hline
\small{CCPT$^{(2)}$} &	$2{N^2}$	&	$2{N^2}-2N$ &	$2{N^2}$	&	$2{N^2}-2N$	\\	\hline
\small{Orthogonal CCPT} & $2{N^2}$	&	$2{N^2}-2N$ & $2Nlog_2(N)-2N+2$	&	$4Nlog_2(N)-7N+10$ \\    \hline
\small{DFT} & 	$4{N^2}$	&	$4{N^2}-2N$ & $2Nlog_2(N)$	&	$3Nlog_2(N)$\\    \hline
\small{RPT} &	$2{N^2}$	&	$2{N^2}-2N$ &	$2{N^2}$	&	$2{N^2}-2N$	\\	\hline
\end{tabular}}
\end{adjustbox}
\normalsize
\end{table}

If $x(n){\in}\mathbb{R}^N$, then DFT coefficients follow the symmetry property, i.e., for each discrete frequency, we have
four real coefficients, whereas CCPT has only two coefficients. This results in a \textit{non-redundant representation}.



\textit{Remark 2:} In literature, there exist methods based on a DIT algorithm for fast $N$-point DFT computation, whenever $N{\neq}2^v$ \cite{oppenheim1975digital,Burrus,Singleton}. In a similar way, extending the proposed DIT-FOCCPT for $N{\neq}2^v$ is one of our future works. 

\textit{Remark 3:} 
There exists a rich literature \cite{winograd,Duhamel,Heideman,4034175}  for fast computation of DFT, with complexity less than DIT-FFT complexity. 
So, further reducing the complexity of OCCPT to less than the DIT-FOCCPT complexity is another open problem, which needs to be addressed.
%

\textit{Notation}: Now onwards, a general CCPT can be either CCPT$^{(1)}$ (or) CCPT$^{(2)}$ (or) OCCPT.
\section{Period and Frequency Estimation}
For a given $N$ length signal, the possible period can be a divisor (or) non-divisor of $N$. In the following subsections, we explore how CCPTs are used to estimate these periods.
\subsection{Divisor Period and Frequency Estimation}
As the CCPT matrix is an NPM, it estimates the period and hidden periods of a signal which are divisors of the signal length.
We illustrate the same by considering two $54$ length periodic sequences $\hat{x}_1(n)$ and $\hat{x}_2(n)$. The period of $\hat{x}_1(n)$ is $18$ and the period of $\hat{x}_2(n)$ is $40$.
In this example, both $\hat{x}_1(n)$ and $\hat{x}_2(n)$ can be periodically decomposed as follows:
\begin{equation}
\hat{x}_1(n) = x_{11}(n)+x_{12}(n)\  \text{and}\ \hat{x}_2(n) = x_{21}(n)+x_{22}(n),
\label{testSignal}
\end{equation}
where $x_{11}(n)$, $x_{12}(n)$, $x_{21}(n)$ and $x_{22}(n)$ are periodic signals with periods $9$, $18$, $5$ and $8$ respectively.
These periods are known as \textit{hidden periods} \cite{7094814}.
Here $x_{11}(n)$, $x_{21}(n)$ are the two periodic random signals, whose one period data is generated by following $\mathcal{N}(0,1)$, $x_{12}(n)=0.6cos\left(2{\pi}100\left(\frac{n}{360}\right)+\frac{\pi}{3}\right)$ and $x_{22}(n)=0.3cos\left(2{\pi}45\left(\frac{n}{360}\right)+\frac{\pi}{4}\right)$. 
Moreover a zero mean Gaussian noise with $SNR=6\ dB$ is added to both $\hat{x}_1(n)$ and $\hat{x}_2(n)$ to generate $x_1(n)$ and $x_2(n)$ respectively.


The three CCPT coefficients are computed for both $x_1(n)$ and $x_2(n)$. Fig. \ref{f1}(a)-(c) and Fig. \ref{f1}(d)-(f) depict the strength of each divisor period present in $x_1(n)$ and $x_2(n)$ respectively.
The strength of a period $p_i$ is computed by taking the square sum of $\varphi(p_i)$ transform coefficients that belong to $s_{p_i}$. 
From Fig. \ref{f1}(a)-(c), the significant periods present in $x_1(n)$ are $3$, $9$ and $18$.
Hence, the period of $x_1(n)$ is equal to $lcm(3,9,18)$.
While from Fig. \ref{f1}(d)-(f), the period of $x_2(n)$ is equal to $54$, as $s_{54}$ has significant period strength.
So, the proposed transforms failed to estimate the period of $x_2(n)$, as
$40{\nmid}54$.

Moreover, CCPT estimates the frequency information of a signal.
Fig. \ref{f1}(g)-(i) depicts the absolute values of transform coefficients that belong to $s_{18}$, computed for $x_1(n)$. 
Note that there are only two significant non-zero coefficients in $s_{18}$.
Because, in CCPT each $\varphi(p_i)$ dimensional $s_{p_i}$ is further decomposed as $\frac{\varphi(p_i)}{2}$ orthogonal CCSs, where each CCS $v_{p_i,k_i}$ consists of signals with frequency $\frac{2{\pi}{k_i}}{p_i}$. 
Now, in this example $s_{18}$ is decomposed as $s_{18} = v_{18,1}{\oplus}v_{18,5}{\oplus}v_{18,7}$. 
Since the sampling frequency is $360 Hz$, the CCSs $v_{18,1}$, $v_{18,5}$ and $v_{18,7}$ consist of signals with frequencies $20 Hz$, $100 Hz$ and $140 Hz$ accordingly. 
As the given $x_{12}(n)$ is a $100 Hz$ signal, the coefficients $15$ and $16$, which belong to ${v_{18,5}}$ are having significant strength. 
Along with this, we can find out the phase information using OCCPT coefficients as given in (\ref{Relation_with_DFT1}).
From Fig. \ref{f1}(i), the two significant coefficient values are $0.149$ and $-0.261$, then 
$-tan^{-1}\left(\frac{-0.261}{0.149}\right) = 1.05\approx \frac{\pi}{3}=1.04\ rad$.
\begin{figure}[!h]
\centering
 \includegraphics[width=5.5in,height=3.4in]{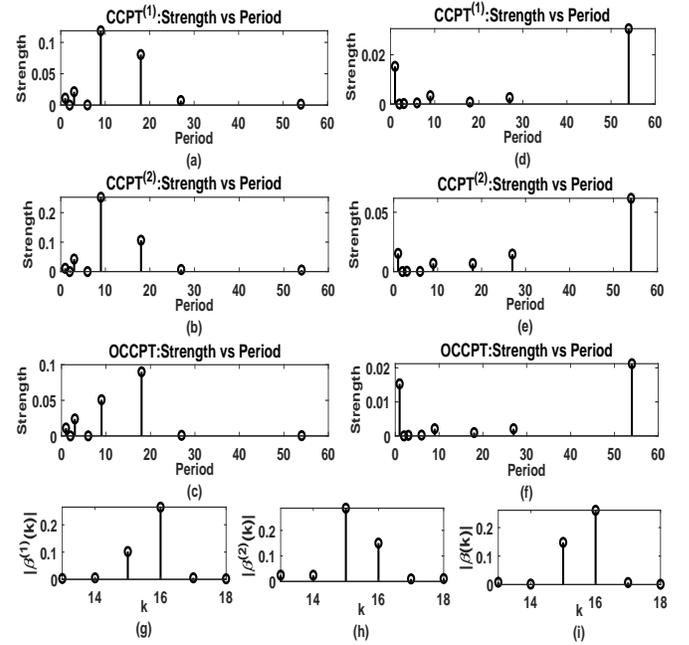} 
\caption[\small{(a)-(c) Strength vs Period plots obtained from the three CCPT coefficients of $x_1(n)$. (d)-(f) Strength vs Period plots obtained from the three CCPT coefficients of $x_2(n)$. (g)-(i) Absolute values of 
CCPT$^{(1)}$, CCPT$^{(2)}$ and OCCPT 
coefficients that belong to $s_{18}$ computed for $x_1(n)$.}]{\small{(a)-(c) Strength vs Period plots obtained from the three CCPT coefficients of $x_1(n)$. (d)-(f) Strength vs Period plots obtained from the three CCPT coefficients of $x_2(n)$. (g)-(i) Absolute values of 
CCPT$^{(1)}$, CCPT$^{(2)}$ and OCCPT 
coefficients that belong to $s_{18}$, computed for $x_1(n)$.}}
\label{f1}
\end{figure}

Therefore, using the proposed transforms we can estimate the divisor period, frequency and phase information of a signal.
Even DFT works for the same purpose, while RPT gives only the divisor period information \cite{Shah}.

Now, we evaluate the period estimation performance of different transforms (DFT, RPT, and CCPTs) in the presence of noise. 
By varying SNRs of $x_1(n)$, we 
compute the strength of each divisor period 
using different transforms.
Consider the periods with significant strength, by keeping $20\%$ of maximum strength as a threshold. 
The obtained results are tabulated in 
Table \ref{tab:EVALUATING THE PERIOD ESTIMATION PERFORMANCE OF DIFFERENT TRANSFORMATION TECHNIQUES IN THE PRESENCE OF NOISE}.
From the table, we claim that the performance of OCCPT is quite good in the presence of noise and it is comparable with DFT and RPT, whereas both CCPT$^{(1)}$ and CCPT$^{(2)}$ are sensitive to noise.
\begin{table}[h]
\centering
\centering
\footnotesize
\caption{E\scriptsize{VALUATING THE PERIOD ESTIMATION PERFORMANCE OF DIFFERENT TRANSFORMATION TECHNIQUES IN THE PRESENCE OF NOISE}}
\label{tab:EVALUATING THE PERIOD ESTIMATION PERFORMANCE OF DIFFERENT TRANSFORMATION TECHNIQUES IN THE PRESENCE OF NOISE}
\begin{adjustbox}{max width=\textwidth}
\renewcommand{\arraystretch}{1.5}
\scalebox{0.67}{
\begin{tabular}{|M{1cm}|M{0.75cm}|M{0.75cm}|M{1.2cm}|M{1.2cm}|M{1.4cm}|M{1.45cm}|M{1.45cm}|} \hline
\bfseries{\small{SNR (dB)}}  
 & \bfseries{\footnotesize{\makecell{6}}} & \bfseries{\footnotesize{\makecell{3}}} & \bfseries{\footnotesize{\makecell{0}}} &
 \bfseries{\footnotesize{\makecell{-3}}} &
 \bfseries{\footnotesize{\makecell{-6}}} &
 \bfseries{\footnotesize{\makecell{-9}}} &
 \bfseries{\footnotesize{\makecell{-12}}} 
 \\ \hline 
 
\small{DFT} & \small{3,9,18}	&	\small{3,9,18}	&	\small{3,9,18} & \small{3,9,18}	&	\small{3,9,18} & \small{3,9,18,27,54}	&	\small{3,9,18,27,54}\\    \hline
\small{RPT} & \small{3,9,18}	&	\small{3,9,18}	&	\small{3,9,18} &	\small{3,9,18}	&	\small{3,9,18}  & \small{1,3,9,18}	&	\small{3,6,9,18,27}	\\	\hline
\small{OCCPT} & \small{3,9,18}	&	\small{3,9,18}	&	\small{3,9,18} & \small{3,9,18}	&	\small{3,9,18} & \small{3,9,18,27,54}	&	\small{3,9,18,27,54}\\    \hline
\small{CCPT$^{(1)}$} & \small{3,9,18}	&	\small{9,18}	&	\small{9,18,27,54} & \small{3,9,18,54}	&	\small{3,9,18,27,54} & \small{9,27,54}	&	\small{9,18,27,54} \\    \hline
\small{CCPT$^{(2)}$} & \small{9,18}	&	\small{9,18}	&	\small{9,18} &	\small{3,9,18,54}	&	\small{9,18,27,54} & \small{9,18,54} 	&	\small{9,27,54}	\\	\hline
\end{tabular}}
\end{adjustbox}
\normalsize
\end{table}

In general, the required period may not be a divisor of the signal length.
In the following subsection, this issue is addressed by considering a dictionary based approach.
\subsection{Non-divisor Period and Frequency Estimation}
In this scenario, the signal is projected onto each and every subspace $s_1$ to $s_{P_{max}}$, instead of projecting only onto the divisor subspaces (as in NPM).
Here $P_{max}$ is the maximum possible period that exists in the signal. 
This generates a fat matrix $\mathbf{F}$ known as the CCPT dictionary. 
Since, $\mathbf{F}$ is fat, there exist multiple solutions ($\mathbf{b}$) for the given signal ($\mathbf{x}$) representation:
\begin{equation}
[\mathbf{x}]_{N{\times}1}=[\mathbf{F}]_{N{\times}\hat{N}}[\mathbf{b}]_{\hat{N}{\times}1},\ \text{where}\ \hat{N} = \sum\limits_{i=1}^{P_{max}}\varphi(i). 
\end{equation} 
In \cite{7109930} and \cite{7094814}, the authors proposed a similar kind of approach for DFT (Farey dictionary) and RPT (RPT dictionary). 
Here the non-divisor period estimation is treated as a 
data fitting problem to reduce the 
computational complexity.
To get the best fit of the given signal with the signals having smaller periods
an optimization problem is formulated as follows:
\begin{equation}
min\ ||\mathbf{T}\mathbf{b}||_2\ \text{s.t.}\ \mathbf{x}=\mathbf{F}\mathbf{b}.
\end{equation}
Here $\mathbf{T}$ is a diagonal matrix consisting of $f(p_i)$ as elements and $p_i$ is the period of $i^{th}$ column in $\mathbf{F}$.
This has a closed-form expression for the optimal solution ($\mathbf{\hat{b}}$) as given below:
\begin{equation}
\mathbf{\hat{b}} = \mathbf{T^{-2}}\mathbf{F^H}(\mathbf{F}\mathbf{T^{-2}}\mathbf{F^H})^{-1}\mathbf{x}.
\end{equation}   
Fig. \ref{f2}(a)-(c) show the strength of each period present in $x_2(n)$ using CCPT dictionaries with $f(p_i) = {p_i}^2$ and $P_{max} = 50$.
For a detailed dictionary approach and for the results of Farey and RPT dictionaries refer \cite{7109930} and \cite{7094814}.
From Fig. \ref{f2}(a)-(c), the period of $x_2(n)$ is equal to $lcm(5,8)$. 

In addition to this, we can estimate the frequency and phase information from CCPT dictionary coefficients by following the same procedure used for CCPT coefficients earlier.
Fig. \ref{f2}(d) shows the absolute values of $\mathbf{\hat{b}}$ obtained using OCCPT dictionary. Here only $70$ coefficients are displayed in the figure, as the rest of the coefficient values are almost equal to zero. 
From Fig. \ref{f2}(d), the significant non-zero coefficient indices $19$ and $20$ are belong to $v_{8,1}$. It indicates the presence of $45 Hz/315 Hz$ frequency component in the signal. These values are $0.0927$ and $-0.0905$, then $-tan^{-1}\left(\frac{-0.0905}{0.0927}\right) = 0.773 \approx \frac{\pi}{4} = 0.785\ rad$.
So, the proposed transforms can be generalized to estimate the non-divisor period, frequency and phase information as well.
\begin{figure}[!h]
\centering
 \includegraphics[width=4.2in,height=2.2in]{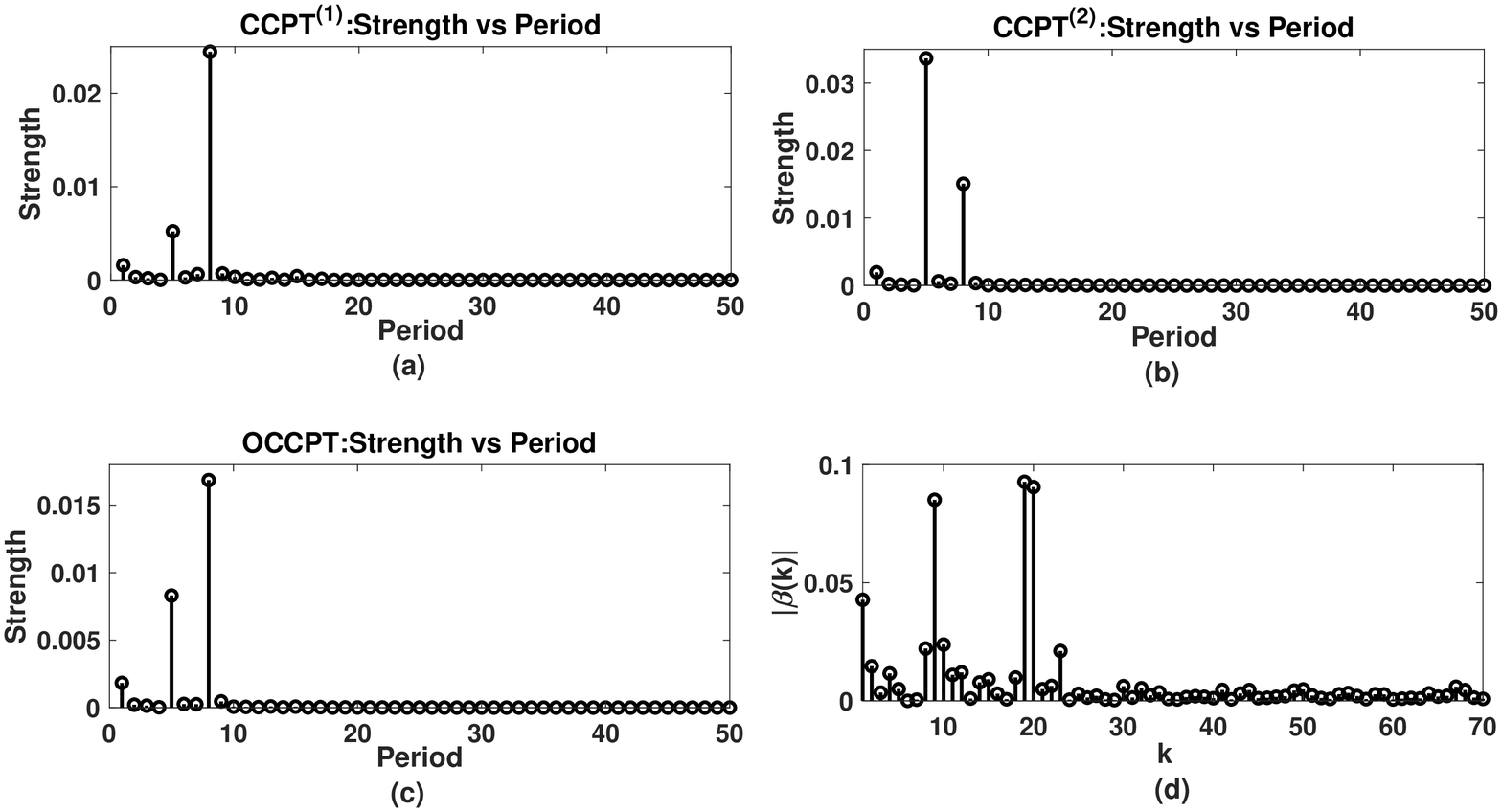} 
\caption[\small{(a)-(c) The strength {\footnotesize vs} period plots  of $x_2(n)$ obtained from the optimal solution ($\mathbf{\hat{b}}$) using CCPT$^{(1)}$, CCPT$^{(2)}$ and OCCPT dictionaries respectively. (d)  Absolute values of optimal solution $\mathbf{\hat{b}}$ computed for $x_2(n)$ using OCCPT dictionary.}]{\small{(a), (b) and (c)- The strength {\footnotesize vs} period plots of $x_2(n)$ obtained from the optimal solution ($\mathbf{\hat{b}}$) using CCPT$^{(1)}$, CCPT$^{(2)}$ and OCCPT dictionaries respectively. (d) Absolute values of optimal solution $\mathbf{\hat{b}}$ computed for $x_2(n)$ using OCCPT dictionary.}}
\label{f2}
\end{figure}
\subsubsection{Analysis of Computational Complexity}
In Farey dictionary $\mathbf{F}{\in}M_{N,\hat{N}}$, the columns of $\mathbf{F}$ follow subspace wise complex conjugate symmetry.
 As a consequence, the \textit{computation of $(\mathbf{F}\mathbf{T^{-2}}\mathbf{F^H})$ results in a real matrix involving complex multiplications.
While for CCPT and RPT dictionaries this computation involves real multiplications}.
Table \ref{tab:Comparison of computational complexity between different dictionaries} tabulates the number of real multiplications and additions required for computing $(\mathbf{F}\mathbf{T^{-2}}\mathbf{F^H})$.
\begin{table}[h]
\centering
\centering
\footnotesize
\caption{C\scriptsize{OMPARISON OF COMPUTATIONAL COMPLEXITY BETWEEN DIFFERENT DICTIONARIES}}
\label{tab:Comparison of computational complexity between different dictionaries}
\begin{adjustbox}{max width=\textwidth}
\renewcommand{\arraystretch}{1.5}
\scalebox{0.75}{
\begin{tabular}{|M{2cm}|M{2.4cm}|M{2.4cm}|M{2.4cm}|} \hline
\bfseries{\small{}}  
 & \bfseries{\footnotesize{\makecell{CCPT \\ Dictionary}}} & \bfseries{\footnotesize{\makecell{Farey \\ Dictionary}}} & \bfseries{\footnotesize{\makecell{RPT \\ Dictionary}}} \\ \hline 
\small{\makecell{Number of\\ Multiplications}} &	${N^2}{\hat{N}}+N\hat{N}$	&	$4{N^2}{\hat{N}}+2N\hat{N}$ &	${N^2}{\hat{N}}+N\hat{N}$	\\	\hline
\small{\makecell{Number of\\ Additions}} &	${N^2}{\hat{N}}+N\hat{N}-{N^2}-N$	&	$4{N^2}{\hat{N}}+2N\hat{N}-2{N^2}-2N$ &	${N^2}{\hat{N}}+N\hat{N}-{N^2}-N$	\\	\hline
\end{tabular}}
\end{adjustbox}
\normalsize
\end{table}
 
If $x(n){\in}\mathbb{R}^N$, then $\mathbf{\hat{b}}$ exhibits subspace wise complex conjugate symmetry for Farey dictionary.
Due to this, the remaining complexity (apart from $(\mathbf{F}\mathbf{T^{-2}}\mathbf{F^H})$ complexity) in computing $\mathbf{\hat{b}}$ is same for both CCPT and Farey dictionaries. 
If $x(n){\in}\mathbb{C}^N$, this symmetry fails, then CCPT and RPT dictionaries have a computational advantage over the Farey dictionary. 
From Table \ref{tab:Comparison of computational complexity between different dictionaries}, the complexity of CCPT and RPT dictionaries is approximately $75\%$ less in comparison with the Farey dictionary.
This computational benefit is also evident from the table given in \cite{7109930}, where the complexity of the Farey dictionary is compared with 
different other dictionaries, which involves real multiplications in computing $(\mathbf{F}\mathbf{T^{-2}}\mathbf{F^H})$ and $\mathbf{\hat{b}}$. 
The overall comparison of the proposed transforms with RPT and DFT is tabulated in TABLE \ref{tab:Comparison of different transformation techniques}.
\begin{table}[h]
\centering
\centering
\footnotesize
\caption{C\scriptsize{OMPARISON OF DIFFERENT TRANSFORMS}}
\label{tab:Comparison of different transformation techniques}
\begin{adjustbox}{max width=\textwidth}
\renewcommand{\arraystretch}{1.5}
\scalebox{0.74}{
\begin{tabular}{|M{1.6cm}|M{1cm}|M{2.2cm}|M{1.4cm}|M{1.2cm}|M{0.7cm}|} \hline
\bfseries{\small{}} & \multirow{2}{3em}{\bfseries{\footnotesize{\makecell{Basis\\ Type\\}}}} & \multirow{2}{8em}{{\bfseries{\footnotesize{\makecell{Orthogonality b/w\\Transformation\\Matrix columns}}}}} & \multirow{2}{5em}{\bfseries{\footnotesize{\makecell{Period\\ Information\\ }}}}& \multicolumn{2}{|c|}{\bfseries{\footnotesize{\makecell{Frequency\\ Information}}}} \\ \cline{5-6}
 & 	&	&  & \bfseries{\footnotesize{Magnitude}} & \bfseries{\footnotesize{Phase}}\\    \hline
\small{DFT} & \small{Complex}	&	\checkmark	&	\checkmark & \checkmark	&	\checkmark\\    \hline
\small{CCPT$^{(1)}$} & \small{Real}	&	$\mathbf{\mathbb{\times}}$	&	\checkmark &	\checkmark	&	$\mathbb{\times}$	\\	\hline
\small{CCPT$^{(2)}$} & \small{Real}	&	$\mathbb{\times}$	&	\checkmark &	\checkmark	&	$\mathbb{\times}$	\\	\hline
\small{OCCPT} & \small{Real}	&	\checkmark	&	\checkmark & \checkmark	&	\checkmark\\    \hline
\small{RPT} & \small{Integer}	&	\checkmark(\small{If matrix size is in power of 2})	&	\checkmark &	$\mathbb{\times}$	&	$\mathbb{\times}$	\\	\hline
\end{tabular}}
\end{adjustbox}
\normalsize
\end{table}
\subsection {Usage of The Proposed Basis in Non-Divisor Subspaces}
The usage of non-divisor subspaces in a signal representation do not hold the orthogonality between the basis elements. 
The proposed bases have computational benefit in such scenarios.
Dictionary based approach discussed above is one such example. Now we mention another example: Given a periodic signal $x(n)$, the minimum data length ($N_{min}$) required to estimate its integer period from a list of candidate integer periods $P = \{P_1,P_2,\dots,P_K\}$ is $N_{min} = \substack{max\\P_i,P_j{\in}P}P_i+P_j-(P_i,P_j)$ \cite{8320848, 7527468}. 
Construct a matrix $\mathbf{H}{\in}M_{N_{min}}(\mathbb{C})$ such that $\mathbf{x} = \mathbf{H}\mathbf{z}$,
where $\mathbf{H}$ includes non-divisor subspaces of $N_{min}$.
For example, let $P = \{6,8\}$, this implies $N_{min} = 12$. Then $\mathbf{H}$ is constructed by using the basis of $s_1, s_2,s_3,s_4,s_6$ and $s_8$. Notice that $s_8$ is not a divisor subspace of $12$. So $\mathbf{H}$ is a non-orthogonal matrix, then $\mathbf{z}=\mathbf{H}^{-1}\mathbf{x}$. One can verify that the $\mathbf{H}$ constructed using complex exponential sequences (or) Ramanujan sums (or) CCPSs is a full rank matrix. In such scenarios using the real-valued CCPSs as basis is computationally efficient for the period and its corresponding frequency estimation over complex exponential sequences.
\subsection{Real-World Example:  ECG Signal Analysis}
Here the problem of R peak (QRS complex) delineation in an ECG signal 
is considered, which is important in many ECG based applications \cite{Pan,Biel}.
We discuss, how to address this problem using DFT, RPT, and CCPTs.
A $10sec$ ECG data with a sampling frequency of $500Hz$ is considered for the analysis (record number $19$ of person $1$ from ECG-ID database \cite{goldbergerphysiobank}).
For easy computation, we further down-sampled this data by a factor of $8$, the resultant $625$ length signal 
is depicted in Fig \ref{f3}(a).
\begin{figure}[!h]
\centering
 \includegraphics[width=4.1in,height=2.2in]{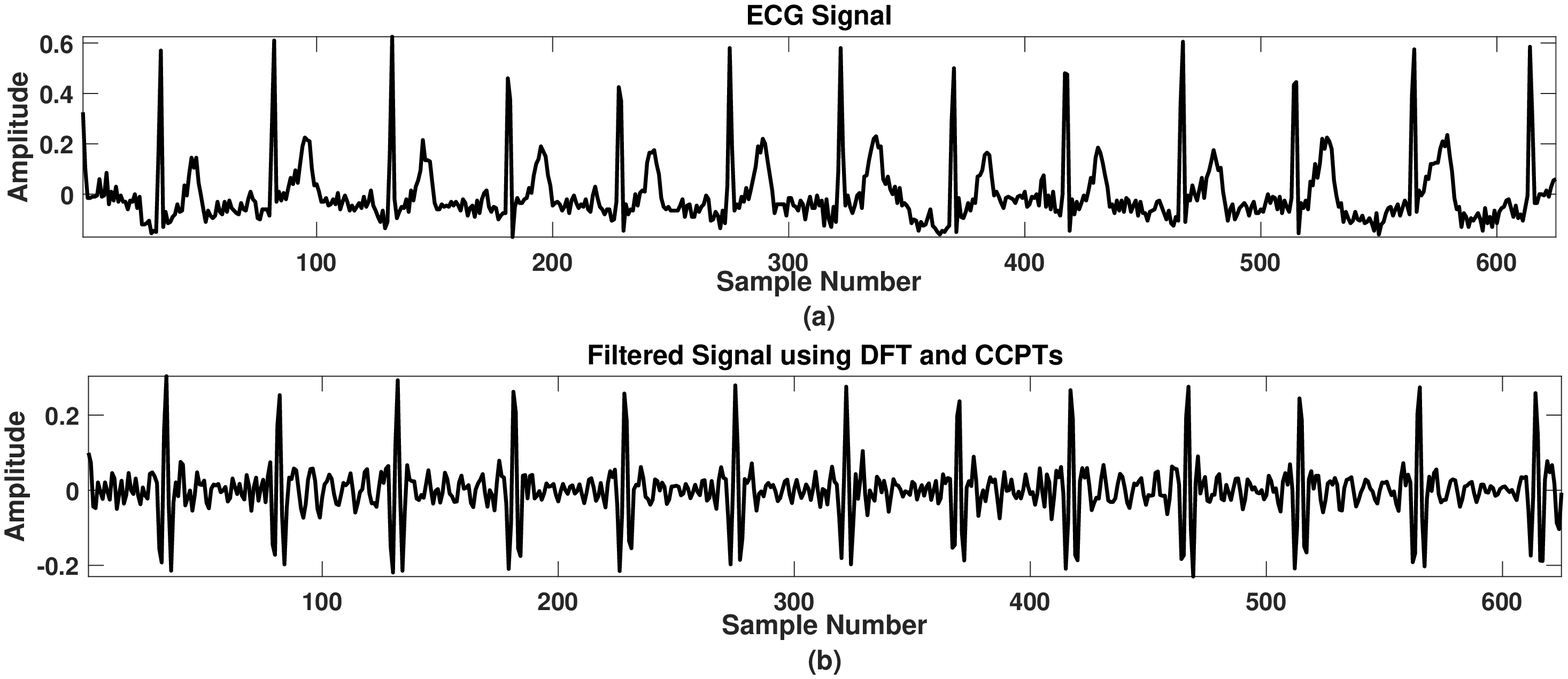} 
\caption[\small{(a) Raw ECG signal. (b) Filtered signal using DFT and CCPTs (CCPT$^{(1)}$, CCPT$^{(2)}$ and OCCPT).}]{(a) Raw ECG signal. (b) Filtered ECG signal using DFT and CCPTs (CCPT$^{(1)}$, CCPT$^{(2)}$ and OCCPT).}
\label{f3}
\end{figure}
Here the period of ECG, i.e.,
the average RR interval is ${0.7733sec \ (\approx48\text{ samples}})$, and $48{\nmid}625$.
As RPT gives only the divisor period information \cite{Shah}, it fails to estimate the R peak locations, whereas CCPTs and DFT give frequency information as well.
Moreover, the frequency range of the QRS complex is $8-20Hz$ \cite{Elgendi}.
So we have reconstructed a signal (filtered) as shown in Fig. \ref{f3}(b), 
by selecting the transform coefficients (of DFT and CCPTs) corresponding to $8-20Hz$ band.
The reconstructed signal from both DFT and CCPTs is the same, since the basis of DFT is orthogonal and in CCPTs the basis is CCS wise orthogonal.
Now, a better estimation of R peak locations can be achieved from this filtered signal, using a standard adaptive threshold algorithm \cite{Pan}.

Addressing the given problem using the dictionary based approach gives the results as shown in Fig. \ref{f4}.
Here we considered $P_{max}=250$ and $f(p_i) = \varphi(p_i)$.
Hence, the period of the ECG signal is $lcm(12,16,48)= 48$, $lcm(12,16)= 48$, $lcm(8,12,16,48)= 48$, $lcm(4,8,12,16)= 48$ and $lcm(8,12,16,48)= 48$ using CCPT$^{(1)}$, CCPT$^{(2)}$, OCCPT, RPT and DFT dictionaries respectively. 
Notice that CCPT$^{(1)}$ and CCPT$^{(2)}$ are giving some spurious periods, as they are more sensitive to noise.
Now, it is easy to estimate the R peak location from the period of the ECG signal. 
\begin{figure}[!h]
\centering
 \includegraphics[width=6.2in,height=2.8in]{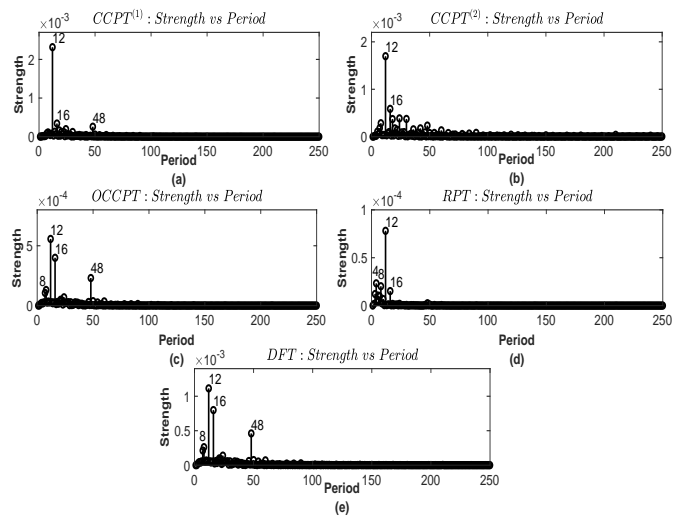} 
\caption[\small{(a)-(e) The strength {\footnotesize vs} period plots of ECG signal obtained from the optimal solution ($\mathbf{\hat{b}}$) using CCPT$^{(1)}$, CCPT$^{(2)}$, OCCPT, RPT and DFT dictionaries respectively.}]{\small{(a)-(e) The strength {\footnotesize vs} period plots of ECG signal obtained from the optimal solution ($\mathbf{\hat{b}}$) using CCPT$^{(1)}$, CCPT$^{(2)}$, OCCPT, RPT and DFT dictionaries respectively.}}
\label{f4}
\end{figure}

By considering the problem of R peak delineation we proved that the period and frequency of an ECG signal can be extracted using the proposed nested periodic matrices/dictionaries.
The DFT matrix/dictionary also gives the same result, but with high computational complexity, whereas the RPT matrix/dictionary gives only the period information of an ECG signal.
\section{Conclusion}
In this paper, we addressed the problem of finite length signal representation by introducing three NPMs.
These NPMs are constructed by providing alternate bases for CCS using CCPSs.
Out of three NPMs, one has mutually orthogonal columns, this results in an orthogonal transform named OCCPT and the remaining two are non-orthogonal transforms.
We proposed two different interpretations for OCCPT such that the information about the period and frequency is explicitly available in each interpretation.
In addition, a DIT based fast computational algorithm is proposed for OCCPT, whenever the length of the signal is equal to $2^v$, $v{\in}\mathbb{N}$. 
Further, we evaluated the performance and computational complexity of the proposed transforms in period and frequency estimation. The results are compared with DFT and RPT.
The proposed theory is justified with some simulated and real-world examples.

\appendix
\textit{Proof of \textbf{Theorem \ref{Th}:}} Let ${L_1}{\geq}3$, ${L_2}{\geq}3$, 
\begin{equation}
\begin{aligned}
E =\sum\limits_{n=0}^{L-1}c_{L_1,k_1}^{(1)}(n-{l_1})c_{L_2,k_2}^{(1)}(n-{l_2}),
\end{aligned}
\label{Orthogonal1}
\end{equation}
$x_1=\frac{2{\pi}{k_1}}{L_1}$ and $x_2=\frac{2{\pi}{k_2}}{L_2}$.
Using Euler's identity and definition of CCPSs, $E$ can be decomposed as
\begin{equation}
\begin{aligned}
E &=\Big[e^{-j({x_1}{l_1}+{x_2}{l_2})}\sum\limits_{n=0}^{L-1}e^{j(x_1+x_2)n}\\&+e^{-j({x_1}{l_1}-{x_2}{l_2})}\sum\limits_{n=0}^{L-1}e^{j(x_1-x_2)n}+e^{j({x_1}{l_1}-{x_2}{l_2})}\\&\sum\limits_{n=0}^{L-1}e^{-j(x_1-x_2)n}+e^{j({x_1}{l_1}+{x_2}{l_2})}\sum\limits_{n=0}^{L-1}e^{-j(x_1+x_2)n}\Big].
\end{aligned}
\label{Orthogonal2}
\end{equation}
Since $L = lcm(L_1,L_2)\ \exists\ {{d_1},{d_2}\in\mathbb{Z}}\ s.t.\ L={L_1}{d_1}\ \text{and}\ L={L_2}{d_2}$. From this
\begin{equation}
\sum\limits_{n=0}^{L-1}e^{{\pm}j({x_1}\pm{x_2})n} 
=\frac{1-e^{{\pm}j2{\pi}({k_1}{d_1}{\pm}{k_2}{d_2})}}{1-e^{\frac{{\pm}j2{\pi}({k_1}{d_1}{\pm}{k_2}{d_2})}{L}}}=0.
\label{Orthogonal3}
\end{equation}
By substituting (\ref{Orthogonal3}) in (\ref{Orthogonal2}), we get, $E=0$. 
The above condition is valid even if $L_1=L_2=L$ and $k_1 {\neq} k_2$.
If $L_1 = L_2 = L$, $k_1 = k_2 = k$  and $l_1{\neq}l_2$, then
\begin{equation}
\sum\limits_{n=0}^{L-1}e^{{\pm}j{({x_1}+{x_2})}n} = 0\ \text{and}\ \sum\limits_{n=0}^{L-1}e^{{\pm}j{({x_1}-{x_2})}n} = L.
\end{equation}
In this case $E=2Lcos\Big(\frac{2{\pi}{k_1}({l_1}-{l_2})}{L_1}\Big)$.
Combining the above cases with
\small{$E = \begin{cases}
	L,& \text{if}\ L_1=L_2=1\text{ (or) }2\\
    0, & \text{if }\ L_1=1\ {\&}\ L_2=2\\
    & \text{ (or) }L_1=2\ {\&}\ L_2=1
\end{cases}$}
leads to
\begin{equation}
E= 2L{M}cos\bigg(\frac{2{\pi}{k_1}({l_1-l_2})}{L_1}\bigg)\delta({L_1}-{L_2})\delta({k_1}-{k_2}).
\end{equation}
\normalsize
\textit{Proof of \textbf{Circular Shift of a Sequence Property:}}
Using 
$c_{p_i,k}^{(1)}\Big(((n-m))_N\Big)$ $=$ $\Big[\frac{1}{2M}c_{p_i,k}^{(1)}(n)c_{p_i,k}^{(1)}\Big(((-m))_N\Big)-\frac{M}{2}c_{p_i,k}^{(2)}(n)c_{p_i,k}^{(2)}\Big(((-m))_N\Big)\Big]$, $c_{p_i,k}^{(2)}\Big(((n-m))_N\Big) = \frac{1}{2M}\Big[c_{p_i,k}^{(2)}(n)c_{p_i,k}^{(1)}\Big(((-m))_N\Big)+c_{p_i,k}^{(1)}(n)c_{p_i,k}^{(2)}\Big((-m)\Big)_N\Big]$ and the synthesis equation given in (\ref{Ortho_CCPT_Synth}), we can write
\begin{equation}
\nonumber
x\Big(((n-m))_N\Big) = \sum_{{p_i}|N} \sum\limits_{\substack{{k}=1\\(k,p_i)=1}}^{\floor*{\frac{p_i}{2}}}\hat{\beta}_{0{k}i}c_{p_i,k}^{(1)}(n)+ \hat{\beta}_{1{k}i}c_{p_i,k}^{(2)}(n),
\end{equation}
where $m{\in}\mathbb{Z}$ and
\begin{equation}
\begin{aligned}
 &\hat{\beta}_{0{k}i}=\frac{1}{2M}\left[{\beta}_{0{k}i}c_{p_i,k}^{(1)}\Big(((-m))_N\Big)+{\beta}_{1{k}i}c_{p_i,k}^{(2)}\Big(((-m))_N\Big)\right],\\
&\hat{\beta}_{1{k}i} = \left[\frac{1}{2M}{\beta}_{1{k}i}c_{p_i,k}^{(1)}\Big(((-m))_N\Big)-\frac{M}{2}{\beta}_{0{k}i}c_{p_i,k}^{(2)}\Big(((-m))_N\Big)\right].
\end{aligned}
\label{Translation_1}
\end{equation}
\normalsize
Now using \footnotesize{$M =\begin{cases}
	\frac{1}{2},& \text{if}\ p_i=1\ \text{(or)}\ 2\\
    1, & \text{if }\ {p_i{\geq}3}
\end{cases}$}, \normalsize and the definitions of CCPSs, we can simplify (\ref{Translation_1}) as given in (\ref{Translation}).

\textit{Proof of \textbf{Circular Convolution Property:}}
Given $x(n) = x_1(n){\circledast}x_2(n)$, then using (\ref{Ortho_CCPT_Analy}) we can write
\begin{equation}
\nonumber
{\beta_{0{k}i}} = \frac{1}{2N{M}}\sum\limits_{n=0}^{N-1}\left[\sum\limits_{l=0}^{N-1}x_1(l)x_2\Big(((n-l))_N\Big)\right]c_{p_i,k}^{(1)}(n).
\end{equation}
Let $n-l = r$, then $\beta_{0{k}i}= \frac{\hat{\beta}_{0ki}}{2M}\textbf{P}-\frac{M\hat{\beta}_{1ki}}{2}\textbf{Q},$
where $\textbf{P}=\sum\limits_{r=-l}^{N-1-l}x_2\Big(((r))_N\Big)c_{p_i,k}^{(1)}(r)=2NM\tilde{\beta}_{0ki}$ and 
$\textbf{Q} = \sum\limits_{r=-l}^{N-1-l}x_2\Big(((r))_N\Big)c_{p_i,k}^{(2)}(r)=2NM\tilde{\beta}_{1ki}$.
Hence
\begin{equation}
{\beta_{0{k}i}} =N\left[\hat{\beta}_{0ki}\tilde{\beta}_{0ki}-{M^2}\hat{\beta}_{1ki}\tilde{\beta}_{1ki}\right].
\end{equation}
Similarly, we can derive
\begin{equation}
{ {\beta_{1{k}i}} = N\left[\hat{\beta}_{0ki}\tilde{\beta}_{1ki}+\hat{\beta}_{1ki}\tilde{\beta}_{0ki}\right].\ }
\end{equation}

Equation (\ref{Conv1}) is an immediate consequence of these equations (obtained by substituting $M$ value).

\textit{Proof of \textbf{Parseval's Relation:}}
Using the orthogonal CCPT synthesis equation, we can write
\begin{equation}
\nonumber
\footnotesize
\begin{aligned}
&\sum\limits_{n=0}^{N-1}|x(n)|^2 = \sum\limits_{n=0}^{N-1}x(n)x^*(n)= \sum_{{p_i}|N} \sum\limits_{\substack{{k}=1\\(k,p_i)=1}}^{\floor*{\frac{p_i}{2}}}\sum_{{p_j}|N} \sum\limits_{\substack{{k_1}=1\\(k_1,p_j)=1}}^{\floor*{\frac{p_j}{2}}}\sum\limits_{n=0}^{N-1}\\
&\Bigg(\underbrace{\beta_{0{k}i}\Big(\beta_{0{k_1}j}\Big)^{*}c_{p_i,k}^{(1)}(n)c_{p_j,k_1}^{(1)}(n)}_{T_1}+
\underbrace{\beta_{0{k}i}\Big(\beta_{1{k_1}j}\Big)^{*}c_{p_i,k}^{(1)}(n)c_{p_j,k_1}^{(2)}(n)}_{T_2}\\
&+\underbrace{\beta_{1{k}i}\Big(\beta_{0{k_1}j}\Big)^{*}c_{p_i,k}^{(2)}(n)c_{p_j,k_1}^{(1)}(n)}_{T_3}
+\underbrace{\beta_{1{k}i}\Big(\beta_{1{k_1}j}\Big)^{*}c_{p_i,k}^{(2)}(n)c_{p_j,k_1}^{(2)}(n)}_{T_4}\Bigg).
\end{aligned}
\normalsize
\end{equation}
Now using \textbf{Theorem \ref{Th3}}, the terms $T_2=T_3=0$ over $0{\leq}n{\leq}N-1$. Similarly, using \textbf{Theorem \ref{Th}}, the terms $T_1=T_4=2NM$, whenever $p_i = p_j$ and $k=k_1$. This implies
\begin{equation}
\sum\limits_{n=0}^{N-1}|x(n)|^2 = \sum_{{p_i}|N} \sum\limits_{\substack{{k}=1\\(k,p_i)=1}}^{\floor*{\frac{p_i}{2}}}2NM\left[|\beta_{0{k}i}|^2+|\beta_{1{k}i}|^2\right].
\end{equation}

\section*{Acknowledgement}
The authors would like to thank Mr. Shiv Nadar, founder and chairman of HCL and the Shiv Nadar Foundation.

\ifCLASSOPTIONcaptionsoff
  \newpage
\fi
\bibliographystyle{IEEEtran}
\bibliography{bibs_1}
%
\begin{IEEEbiography}[{\includegraphics[width=1in,height=1.2in]{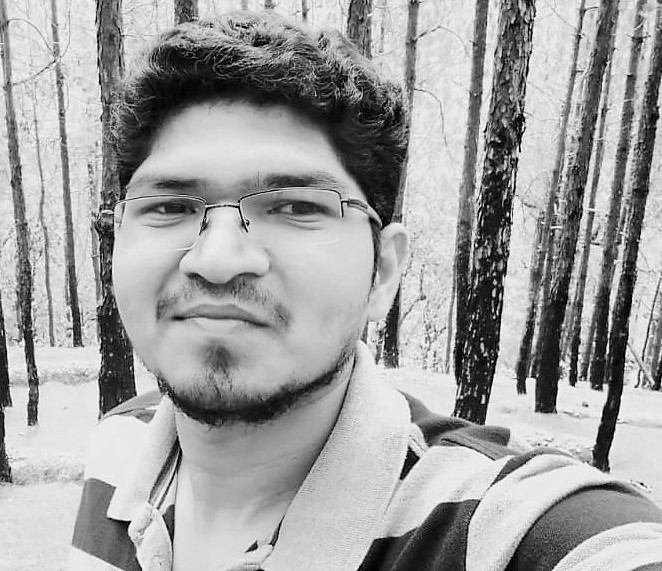}}]{Shaik Basheeruddin Shah}(S'18)
has a B.Tech degree in Electronics and Communication Engineering from the Vasireddy Venkatadri Institute of Technology (VVIT), in 2013 and a Master's degree in Computational Engineering from the Rajiv Gandhi University of Knowledge and Technologies (RGUKT) in 2015. Currently, he is a Ph.D. student at Shiv Nadar University, India. His research interests lie in the area of Discrete-time Signal Representation and Analysis.
\end{IEEEbiography}
\begin{IEEEbiography}[{\includegraphics[width=1in,height=1.2in]{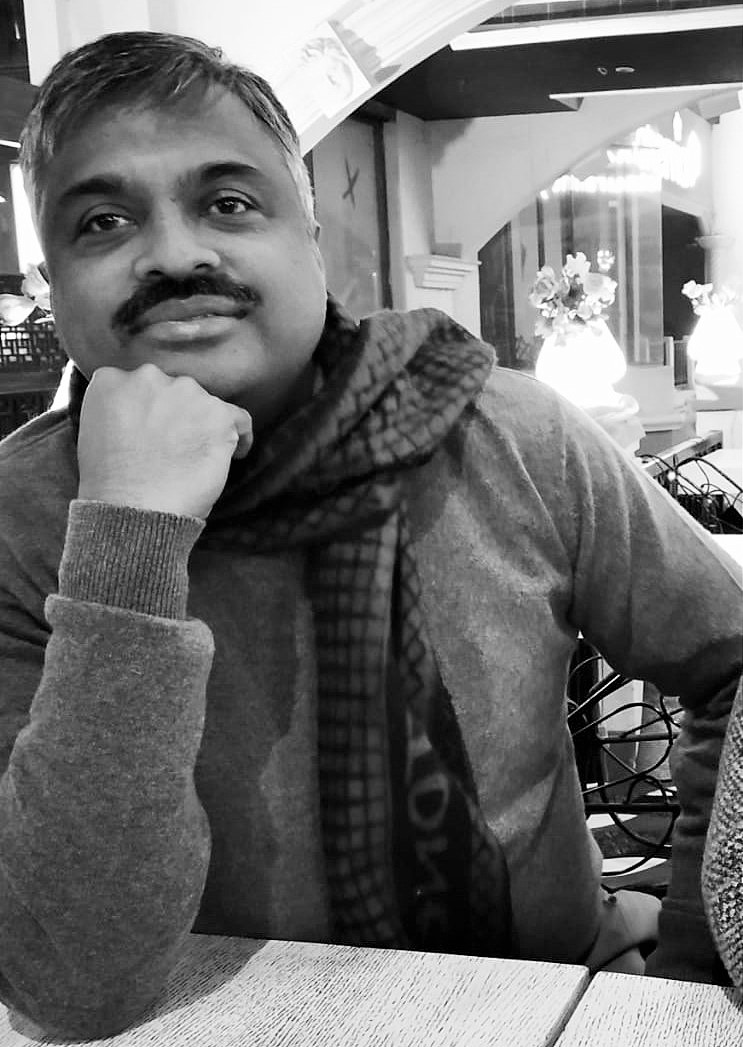}}]{Vijay Kumar Chakka}(M'09-SM'12) 
received B.Tech, M. Tech, Ph.D., in ECE from JNTU Hyderabad, National Institute of Technology,  Kurukshetra, India and  from National Institute of Technology, Trichy,   India on 1991,1993 and 2004 respectively.
He is currently working as a Professor at Shiv Nadar University, Greater Noida, India from 2014 onwards. He was Associate Professor at Dhirubhai Ambani Institute of Information and Communication Technology, Gandhinagar from 2002-2013. He was also an Adjunct Associate professor at IIT Gandhinagar from 2010-12. He was secretary to the IEEE communication chapter of IEEE Gujarat section from 2010-12. Before joining DA-IICT, he worked as a Senior Lecturer at National Institute of Technology, Trichy, India from 1994-2002. He worked as consultant for DRDO India, NSTL Vizag and many private companies in the area of Digital Signal Processing and Wireless Communication.  His research interests are Signal representation, Signal design for 5G and pre coders designs for wireless communication etc.
\end{IEEEbiography}
\begin{IEEEbiography}[{\includegraphics[width=1in,height=1.2in]{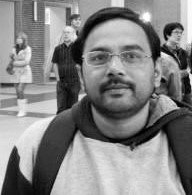}}]{Arikatla Satyanarayana Reddy}
has received B.Sc from the Ideal Degree College, Kakinada in 1993, an M.Sc-Mathematics, MA-Education from Andhra University, Vishakhapatnam in 1995, 1999 respectively. He received a Ph.D. degree in the Department of Mathematics and Statistics from the Indian Institute of Technology (IIT), Kanpur, India in 2012. Currently, he is working as an Associate Professor in the Department of Mathematics, Shiv Nadar University, India.
His research fields of expertise include Algebraic Graph Theory, Linear Algebra and Algebraic Number Theory.
\end{IEEEbiography}

%





\end{document}